%% file: ArXiv.tex
\documentclass[12pt]{article}
\usepackage[utf8]{inputenc} 
\usepackage[T1]{fontenc} 
\pdfoutput=1

\usepackage[margin=1in]{geometry}
\usepackage{fullpage}
\usepackage{amsmath,amssymb, amsthm}
\newtheorem{theorem}{Theorem}
\newtheorem{lemma}{Lemma}
\newtheorem{proposition}{Proposition}

\usepackage{enumerate}
\usepackage{enumitem}

\usepackage{color-edits}
\addauthor{ao}{magenta}     
\addauthor{yl}{blue}    

\input{math}

\input{notation}

\usepackage[capitalize]{cleveref}
\usepackage{hyperref}

\begin{document}

\title{Dynamics and Contracts for an Agent with Misspecified Beliefs\thanks{The authors thank Ryota Iijima for helpful comments and suggestions. Yingkai Li acknowledges financial support from the Sloan Research Fellowship, grant no.~FG-2019-12378. Argyris Oikonomou acknowledges financial support from a Meta PhD fellowship, a Sloan Foundation Research Fellowship and the NSF Award CCF-1942583 (CAREER).}}

\author{Yingkai Li\thanks{
Cowles Foundation for Research in Economics and Department of Computer Science, Yale University.
Email: \texttt{yingkai.li@yale.edu}}
\and Argyris Oikonomou\thanks{
Department of Computer Science, Yale University.
Email: \texttt{argyris.oikonomou@yale.edu}}
}

\date{}

\maketitle

\begin{abstract}
We study a single-agent contracting environment where the agent has misspecified beliefs about the outcome distributions for each chosen action. First, we show that for a myopic Bayesian learning agent with only two possible actions, the empirical frequency of the chosen actions converges to a Berk-Nash equilibrium. However, through a constructed example, we illustrate that this convergence in action frequencies fails when the agent has three or more actions. Furthermore, with multiple actions, even computing an $\epsilon$-Berk-Nash equilibrium requires at least quasi-polynomial time under the Exponential Time Hypothesis (ETH) for the PPAD-class. This finding poses a significant challenge to the existence of simple learning dynamics that converge in action frequencies. Motivated by this challenge, we focus on the contract design problems for an agent with misspecified beliefs and two possible actions. We show that the revenue-optimal contract, under a Berk-Nash equilibrium, can be computed in polynomial time. Perhaps surprisingly, we show that even a minor degree of misspecification can result in a significant reduction in optimal revenue. 
\end{abstract}

\input{content/intro}
\input{content/prelim}

\input{content/learning_eql}

\input{content/berk_nash_ppad}

\input{content/efficiency}
\input{content/lower_bound}
\input{content/conclusion}

\bibliographystyle{apalike}
\bibliography{ref}
\newpage
\appendix
\input{appendix/eqn_learning}

\input{content/divergence_under_mispecification}
\input{appendix/opt_contract}

\input{appendix/unhappy_principal}
\end{document}

%% file: math.tex
\newcommand{\Nat}{\mathbb{N}}
\newcommand{\ind}[1]{\mathbbm{1}\left[#1\right]}
\newcommand{\InBraces}[1]{\left\{  #1  \right\}}

\newcommand{\InBrackets}[1]{\left[  #1  \right]}

\newcommand{\inr}[1]{\left\langle #1 \right\rangle}

\DeclareMathOperator*{\argmax}{arg\,max}
\DeclareMathOperator*{\argmin}{arg\,min}
\DeclareMathOperator*{\poly}{poly}
\DeclareMathOperator*{\polylog}{polylog}

\newcommand{\given}{|}

\newcommand{\prob}[2][]{\text{\bf Pr}\ifthenelse{\not\equal{}{#1}}{_{#1}}{}\!\left[{\def\givenn{\middle|}#2}\right]}
\newcommand{\expect}[2][]{\text{\bf E}\ifthenelse{\not\equal{}{#1}}{_{#1}}{}\!\left[{\def\givenn{\middle|}#2}\right]}

\newcommand{\tparen}{\big}
\newcommand{\tprob}[2][]{\text{\bf Pr}\ifthenelse{\not\equal{}{#1}}{_{#1}}{}\tparen[{\def\given{\tparen|}#2}\tparen]}
\newcommand{\texpect}[2][]{\text{\bf E}\ifthenelse{\not\equal{}{#1}}{_{#1}}{}\tparen[{\def\given{\tparen|}#2}\tparen]}

\newcommand{\sprob}[2][]{\text{\bf Pr}\ifthenelse{\not\equal{}{#1}}{_{#1}}{}[#2]}
\newcommand{\sexpect}[2][]{\text{\bf E}\ifthenelse{\not\equal{}{#1}}{_{#1}}{}[#2]}
\newcommand{\BR}[1]{\textsc{BR}\InParentheses{#1}}

\newcommand{\reals}{\mathbb{R}}

\newcommand{\InParentheses}[1]{\left({#1}\right)}
\newcommand{\InAbsolute}[1]{\left|{#1}\right|}
\newcommand{\KL}[2]{{\rm D_{KL}}\InParentheses{#1\parallel #2}}

\newcommand{\rbr}[1]{\left(#1\right)}

\renewcommand{\mod}[1]{\InParentheses{\text{mod $#1$} }}

\newcommand{\uceil}[1]{\left \lceil #1 \right \rceil }

\usepackage{changepage}
\usepackage{enumerate}
\usepackage{enumitem}
\usepackage{scalerel}
\usepackage{accents}
\usepackage{bbm}
\usepackage{hyperref}
\usepackage{xfrac}

\usepackage{thm-restate}

\usepackage{natbib}
\usepackage[capitalize]{cleveref}
\usepackage{dsfont}
\usepackage{tikz}
\usetikzlibrary{patterns}

\newenvironment{prevproof}[1]{\noindent {\bf {Proof of \Cref{#1}:}}}{$\hfill \blacksquare$}
\newtheorem{definition}{Definition}

\newtheorem{instance}{Instance}
\newtheorem{assumption}{Assumption} 
\newtheorem{hypothesis}{Hypothesis} 
\newtheorem{observation}{Observation}

\newcommand{\notshow}[1]{}

%% file: notation.tex
\newcommand{\actions}{A}
\newcommand{\action}{a}
\newcommand{\transfer}{p}
\newcommand{\contract}{P}
\newcommand{\contracts}{\mathcal{P}}
\newcommand{\cost}{C}
\newcommand{\util}{v}
\newcommand{\rev}{u}
\newcommand{\Util}{V}
\newcommand{\Rev}{U}
\newcommand{\reward}{r}
\newcommand{\rewards}{R}
\newcommand{\dist}{F}

\newcommand{\belief}{B}
\newcommand{\beliefs}{\mathcal{B}}
\newcommand{\kl}{{\rm KL}}
\newcommand{\kls}{\mathcal{K}}
\newcommand{\prior}{\rho}
\newcommand{\posterior}{\mu}
\newcommand{\adist}{\alpha}

\newcommand{\supp}{{\rm supp}}

\usepackage{adjustbox}

%% file: content/intro.tex
\section{Introduction}
\label{sec:intro}

Contract design is a cornerstone of economic theory, playing a pivotal role in a wide range of applications where interests are misaligned, such as in insurance markets, employment contracts, corporate governance, and venture capital financing. In these applications, the actions undertaken by the agent are not observed by the principal, which necessitates the need of contracts that reward the agent based only on the observable outcomes influenced by their actions. Contract design theory revolves around the provision of optimal incentives for agents to undertake actions that, although unobservable, are beneficial to the principal.

The classic model of contract design assumes common knowledge between the principal and the agent concerning the distribution of outcomes generated by each private action chosen by the agent. However, in practical applications, the agent may have misspecified beliefs about the outcome distributions. For instance, the agent might be overly optimistic about the impact of their effort on the firm's output, thus overestimating the distribution of outputs when exerting costly effort. Such misspecification could significantly affect the agent's behavior under any given contract and, subsequently, influence the design of optimal contracts for maximizing the principal's revenue.

When the agent has misspecified beliefs in contracting environments, the classic notion of Bayesian-Nash equilibrium is not suitable for capturing equilibrium behavior. Recently, \citet{esponda2016berk} introduced the concept of Berk-Nash equilibrium for agents with misspecified beliefs. Intuitively, under a Berk-Nash equilibrium, the agent selects subjective beliefs that minimize the Kullback-Leibler divergence from the observed outcome distributions given his actions, and these actions are the best responses to his chosen beliefs. This definition aligns with the concept of self-confirming equilibrium, as described by \citet{Fudenberg93self}, in enironments where the agent's beliefs are correctly specified.\footnote{Note that in contract design environments, the agent only observes the outcomes associated with the chosen action. } 

\citet{esponda2016berk} also provide a learning foundation for the Berk-Nash equilibrium concept when the agent has misspecified beliefs. Specifically, in a repeated contracting environment where the agent is offered the same contract across all time periods, the agent myopically chooses an action in each period that maximizes his single-period utility based on his posterior belief updated using past observations. \citet{esponda2016berk} connects Berk-Nash equilibrium with the learning model by showing that if the learning dynamics chosen by the myopic learning agent converge in a perturbed environment, they must converge to a Berk-Nash equilibrium of that perturbed environment. However, convergence under misspecification cannot always be guaranteed \citep[e.g.,][]{nyarko1991learning,esponda2021asymptotic}.

In our paper, we first show that when only two possible actions exist, the action frequencies of a myopic Bayesian learning agent converge, with the limiting point being a Berk-Nash equilibrium. This provides strong justification for using Berk-Nash equilibrium as the solution concept for contract design problems involving two actions. However, when three or more actions are available, the action frequencies of a myopic Bayesian learning agent may not converge. A natural question arises regarding the existence of other simple learning dynamics that could justify the use of Berk-Nash equilibrium. However, we show that under the Exponential Time Hypothesis (ETH) for the PPAD class, computing a Berk-Nash equilibrium generally requires at least quasi-polynomial time when there are three or more actions. This suggests that finding simple learning dynamics that converge to a Berk-Nash equilibrium within polynomial time is not possible without additional assumptions on the game structures.

Motivated by the challenges in environments with three or more actions, our paper focuses on the design of optimal contracts with two available actions. We show that designing optimal contracts can be formulated as a quadratic program, and we provide a polynomial-time algorithm for solving such a program by leveraging the structure of feasible contracts and Berk-Nash equilibria with only two actions. Intuitively, to compute the optimal contract, we enumerate the support of actions and beliefs sustainable under the Berk-Nash equilibrium. We then utilize four different linear programs to solve for the optimal contracts that maximize the principal's revenue, subject to the constraint that the induced Berk-Nash equilibrium is supported within the enumerated set.

Finally, we show that in contract design environments with binary actions, even a minor degree of misspecification can lead to a significant change in the format of the optimal contract and a substantial decrease in optimal revenue compared to the correctly specified model. Our findings highlight the importance of robustness in economic models of contract design. Given that slight misspecifications can result in substantially different outcomes, this prompts further research into robust contract design that ensures high revenue and exhibits less sensitivity to errors in the agent's specification. We present this as an intriguing open question in our paper.



\subsection{Related Work}
\label{sub:literature}
The contract design problems are well studied in classic Bayesian environments without misspecification \citep[e.g.,][]{Holmstrom79,GrossmanHart83}. 
In recent years, there is a growing literature on the computational complexity of finding the optimal contracts in such Bayesian environments. 
When agents have private types concerning their cost of effort, or the distribution over outcomes given each private action, 
in single-dimensional private type settings, \citet{AlonDT21} show that given a constant number of actions, the optimal deterministic contract can be computed in polynomial time.
In multi-dimensional private type settings, 
\citet{guruganesh20} show that even for a constant number of actions, computing the optimal deterministic contract is NP-hard. 
In contrast, \citet{castiglioni2022designing} show that for arbitrary number of actions, the optimal randomized contract can be computed in polynomial time. 
In models without private types, \citet{DuttingEFK21,DuettingEFK23} consider the computation of (approximately) optimal contracts in combinatorial environments with multiple actions or multiple agents. 


In models where the agents have misspecified beliefs, the seminal paper by \citet{esponda2016berk} proposed the equilibrium concept of Berk-Nash equilibrium, which includes classic self-confirming equilibrium \citep{Fudenberg93self} and Bayes-Nash equilibrium as special cases. They also provide a learning foundation for Berk-Nash equilibrium in games with finitely many actions by showing that in a perturbed environment, if the actions of myopic Bayesian learning agents converges, 
the limiting point must be a Berk-Nash equilibrium. 
\citet{fudenberg2021limit} further show that for environments with a continuum of actions, only uniform Berk-Nash equilibrium can long-run outcomes. 
However, in learning environments with misspecified beliefs, as illustrated in \citet{nyarko1991learning,fudenberg2017active}, convergence in actions is not guaranteed. 
\citet{esponda2021asymptotic} provide sufficient conditions for the convergence in action frequencies. 
Our results complement their findings by showing convergence in binary action environments without additional assumptions. 
Moreover, \citet{frick2023belief} show that efficient social learning fails even with a small amount of misspecification, a finding that echoes our result, where even a minimal degree of misspecification can lead to significant revenue loss in contracting environments.

Our paper is among one of the few papers that study the contract design problems when agents have misspecified beliefs. 
\citet{li2020misspecified} consider the optimal dynamic contract design when the agent has misspecified beliefs regarding the time lags of the observed outcomes for each chosen action. 
\citet{li2022incentivizing} consider the design of clinical trials to incentivize the patients to participate in the trials even when patients have misspecified beliefs regarding the treatment outcomes. 
Both the models and the design spaces in these papers are significantly different from our paper. 
Finally, a recent paper by \citet{guruganesh2024contracting} studies the dynamic contract design problem when the agent is not myopic and is using no-regret learning algorithms for responding to the designed contracts. 
In contrast, we focus on the design of static contracts for myopic Bayesian learning agents with misspecified beliefs, which leads to a significantly different form of optimal contract.

Our paper also connect to learning in games without misspecification. When two no-regret learning agent (e.g. agents using gradient descent) interact in a bimatrix zero-sum games, their action-profile might cycle in the space perpetually, despite their empirical frequency converging to a Nash-equilibrium \citep{MertikopoulosPP18}. Two notable stabilized training methods include Optimistic Gradient Descent by \cite{popov_modification_1980} and the Extragradient method by \cite{korpelevich_extragradient_1976}. These two variants are known to converge to a Nash-equilibrium in zero-sum games \cite{popov_modification_1980,korpelevich_extragradient_1976,hsieh_convergence_2019,facchinei_finite-dimensional_2007}.

%% file: content/prelim.tex
\section{Preliminaries}
\label{sec:prelim}

We consider a general contract design model with moral hazard. 
In this model, the agent can privately choose a costly action from a finite set $\actions$ 
that affects the distribution over reward for the principal. 
Specifically, for any action $\action\in\actions$, 
the cost of action $\action$ is $\cost(\action)$
and a publicly observed reward $\reward\in\rewards$ for the principal is drawn according to distribution~$\dist_a$. 
We assume that $\rewards\subseteq\reals$ is a finite set.

\paragraph{Utilities}
We assume both the principal and the agent are risk neutral and have quasi-linear utilities. 
That is, given transfer $\transfer$ from the principal to the agent, 
the utility of the agent for choosing action $\action$ is 
\begin{align*}
\util(\action,\transfer) &= \transfer-\cost(\action)
\intertext{and the revenue of the principal for realized reward $\reward$ is} 
\rev(\reward,\transfer) &= \reward - \transfer.
\end{align*}

\paragraph{Contracts}
A contract rule $\contract:\rewards\to \reals_+$ is a mapping from contractable rewards to payments. 
Note that we impose the limited liability constraint such that the payments to the agent are non-negative \citep{innes1990limited}.
In this model, the principal can commit to a contract rule $\contract$ to incentivize the agent to choose costly actions. 
Upon observing the contract, 
the agent privately chooses an action $\action\in \actions$
and then a public reward $\reward$ is realized according to distribution~$\dist_{\action}$.
The principal pays price $\contract(\reward)$ to the agent. 
Given contract $\contract$, the expected utility of the principal and the agent when action $\action$ is chosen 
are  
\begin{align*}
\Util(\action,\contract) &= \expect[\reward\sim\dist_{\action}]{\contract(\reward)}-\cost(\action)
\intertext{and}
\Rev(\action,\contract) &= \expect[\reward\sim\dist_{\action}]{\reward - \contract(\reward)}
\end{align*}
respectively. We also omit $\contract$ in the notation when it is clear from the context.

\subsection{Subjective Beliefs and Misspecification}
\label{sub:beliefs}
In our model, the true data generating process (reward distribution) $\dist_{\action}$ is unknown to the agent in ex ante, 
and the agent holds a subjective belief over various possible data generating processes. 
Specifically, let $\belief$ denote a possible belief of the reward distribution where $\belief_{\action} \in \Delta(\rewards)$ 
is the subjective distribution over rewards for action $\action$ 
given belief $\belief$. 
The expected utility of the agent given subjective belief $\belief$
under contract $\contract$ for choosing action~$\action$ is 
\begin{align*}
\Util(\action,\belief,\contract) &= \expect[\reward\sim\belief_{\action}]{\contract(\reward)}-\cost(\action)
\end{align*}
and we also omit $\contract$ in the notation when it is clear from the context.

In this subjective belief model, the agent has a prior $\prior\in\Delta(\beliefs)$ over the set of possible beliefs. 
A prior $\prior$ is misspecified if the true reward distribution $\dist=\InBraces{\dist_a}_{a\in\actions}$ is not in the support of $\prior$. 
We impose a generic assumption on the agent's misspecified beliefs based on the notion of KL-divergence that measures the distances among distributions. 
Specifically, for any pair of distributions $F, \hat{F}\in\Delta(\rewards)$, 
the KL-divergence between $F$ and $\hat{F}$ is 
\begin{align*}
\KL{F}{\hat{F}} = \expect[\reward\sim F]{\log\rbr{\frac{F(\reward)}{\hat{F}(\reward)}}}. 
\end{align*}
In Sections \ref{sec:learning agent two actions} and \ref{sec:polytime md for two actions}, for environments with two possible actions, we make an additional assumption that the agent's beliefs are generic.
At a high level, \Cref{ass:poly contract} posits that for an agent with two actions, any pair of belief have distinct KL divergence and the KL divergence induced by any three of her beliefs cannot be colinear. This assumption can be satisfied by a random perturbation of the agent's beliefs.
\begin{assumption}[Generic Beliefs over Two Actions]
\label{ass:poly contract} 
In environments with two possible actions, i.e., $|\actions|=2$, 
the agent has generic misspecified beliefs if:
\begin{itemize}

\item For any action $\action \in \actions$ and two different beliefs $\belief\neq\belief'\in \beliefs$:
\begin{align*}
\KL{\dist_\action}{\belief_\action}\neq\KL{\dist_\action}{\belief'_\action}.
\end{align*} 

\item For a given belief \(\belief \in \beliefs\), we define the Kullback-Leibler divergence set \(\kl(\belief)\) as
$$
\kl(\belief) = \left\{ \KL{\dist_\action}{\belief_\action} \right\}_{\action \in \actions} \subseteq \mathbb{R}^{2}.
$$
The set of such diverges is denoted by
$\kls = \left\{ \kl(\belief) \right\}_{\belief \in \beliefs}$, and is in general-position.\footnote{That is, for any triplet of distinct diverges of beliefs $\kl(\belief_0)\neq\kl(\belief_1)\neq\kl(\belief_2)\in\kls$, vectors $\kl(\belief_0)-\kl(\belief_1)$ and $\kl(\belief_0)-\kl(\belief_2)$ are linearly independent.}
\end{itemize}
\end{assumption}


\subsubsection{Berk-Nash Equilibrium}\label{sec:definition of Berk-Nash}
In models with misspecified beliefs, a well-known solution concept proposed by \citet{esponda2016berk} is Berk-Nash equilibrium. 
Given contract $\contract$, 
a distribution $\adist\in\Delta(\actions)$ over actions forms a Berk-Nash equilibrium if there exists 
a posterior $\posterior\in\Delta(\beliefs)$ such that 
\begin{itemize}
\item \emph{Optimality.}
For any action $\action^*\in \supp(\adist)$,\footnote{$\supp(\cdot)$ is the support function of a distribution.} 
action $\action^*$ maximizes the agent's expected utility given posterior $\posterior$, i.e., 
\begin{align*}
\action^* \in \argmax_{\action\in\actions} \expect[\belief\sim\posterior]{\Util(\action,\belief,\contract)}.
\end{align*}

\item \emph{Consistency.}
For any belief $\belief^*\in \supp(\posterior)$ in the support of the posterior $\posterior$, 
belief $\belief^*$ minimizes the expected KL-divergence among all beliefs in the support of the prior, 
i.e., 
\begin{align*}
\belief^* \in \argmin_{\belief\in \beliefs} 
\expect[\action\sim\adist]{\KL{\dist_{\action}}{\belief_{\action}}}.
\end{align*}
\end{itemize}
Note that for any contract $\contract:\rewards\rightarrow\reals_+$, the existence of a Berk-Nash equilibrium is promised by \citet{esponda2016berk}.

\begin{proposition}[\citealp{esponda2016berk}]
    For any contract $\contract:\rewards\rightarrow \reals_+$, there exists at least one Berk-Nash equilibrium.
\end{proposition}

In our paper, we also consider a relaxed notion of Berk-Nash equilibrium by allowing the agents to make an $\epsilon$ mistake for choosing the optimal action. 
\begin{definition}[$\epsilon$-Berk-Nash equilibrium]\label{def:eps Berk-Nash}
For any $\epsilon\geq 0$, a distribution $\adist\in\Delta(\actions)$ over actions forms an $\epsilon$-Berk-Nash equilibrium if there exists a posterior $\posterior\in\Delta(\beliefs)$ 
such that both consistency and the following condition are satisfied:
\begin{itemize}
\item \emph{$\epsilon$-optimality.} 
The distribution $\adist$ over actions approximately maximizes the agent's expected utility given posterior $\posterior$, i.e., 
\begin{align*}
\expect[\action^*\sim\adist]{\expect[\belief\sim\posterior]{\Util(\action^*,\belief,\contract)}}
\geq \max_{\action\in\actions} \expect[\belief\sim\posterior]{\Util(\action,\belief,\contract)} - \epsilon.
\end{align*}
\end{itemize}
\end{definition}

\subsubsection{Bayesian Learning Agent}\label{sec:learning agent}
Another way to analyze the behavior of the agent in a misspecified agent is to assume that the agent is a myopic Bayesian learning agent who interacts with the contract environment for multiple periods. 
Specifically, the belief of the agent at time $t=0$ is his prior, i.e., $\posterior_0 = \prior$. 
For any time $t\geq 1$, 
the agent chooses an action 
\begin{align*}
\action_t\in \argmax_{\action\in\actions} \expect[\belief\sim\posterior_{t-1}]{\Util(\action,\belief,\contract)}.
\end{align*}
Then the reward at time $t$ is realized according to the true reward distribution $\dist_{\action_t}$, 
and the agent updates his posterior $\posterior_t$ according to Bayes rule.

%% file: content/learning_eql.tex
\section{Convergence to Berk-Nash for a Learning Agent with Two Actions}\label{sec:learning agent two actions}
In this section, we show that for a learning agent with two actions, her action frequency converges to a Berk-Nash equilibrium of the game. Formally, we consider a principal who commits to a fixed contract rule $\contract$, and an agent who chooses from two possible actions $\actions=\InBraces{\action_1,\action_2}$ and has a set of misspecified beliefs $\beliefs$ about the outcome of her actions that satisfy the generic belief assumption (\Cref{ass:poly contract}). We assume that the agent updates her belief using the Bayes rule  and chooses the action that maximizes her utility based on her posterior (see \Cref{sec:learning agent}). 
To simplify the exposition, we assume that the agent has a unique best-response for each belief $\belief \in \beliefs$. 
This assumption is summarized in \Cref{as:best response bayes}. 
In \Cref{lem:best response bayes} in \cref{subapx:convergence_binary}, we also show that making this assumption is without loss of generality, as beliefs where the agent is indifferent between the two actions do not affect the dynamics. We postpone the proofs of this section to \Cref{subapx:convergence_binary}. 

Convergence results for learning agents are considered in Section~7 by \citet{esponda2021asymptotic} for an agent with multiple actions and structured one-dimensional belief spaces. In a high-level, they show convergence either when i) for each distribution over actions $\adist\in \Delta(\actions)$ there exists a unique belief that minimizes the KL-divergence with respect to $\adist$ or ii) when the agents picks pure action $\action\in \actions$, the KL-divergence of the beliefs monotonically changes as the belief ``increase'' within it's one-dimenstional space.\footnote{In a high-level, since the belief space is one-dimentional, we can think of the belief space as the real line.} Interested readers can see Section~7 in \citet{esponda2021asymptotic} for more details. Our convergence result depart from prior work by considering agent with two actions, and unstructured beliefs. 
We only impose the generic belief assumptions as in \Cref{ass:poly contract} and \Cref{as:best response bayes}.

\begin{assumption}\label{as:best response bayes}
    For each $\belief \in \beliefs$, there exists a unique best-response, e.g. $\expect[\reward\sim \belief_{\action_1}]{r}\neq\expect[\reward\sim \belief_{\action_2}]{r}$.
\end{assumption}

First we introduce the following notation:
\begin{itemize}
    \item For any distribution over actions $\adist\in \Delta(\actions)$, let $$\beliefs\InParentheses{\adist}=\argmin_{\belief\in\beliefs}\expect[\action\sim\adist]{\KL{\dist_{\action}}{\belief_{\action}}},$$ be the beliefs that minimize the KL-divergence when the action is sampled by $\adist$.



    \item For a distribution over actions $\adist \in \Delta(\actions)$, we denote best-responding actions over posteriors supported on $\beliefs(\adist)$ by: $$\BR{\adist}=\InBraces{\action^*\in \actions: \exists \posterior\in \Delta(\beliefs(\adist)): \action^* \in \argmax_{\action\in\actions} \expect[\belief\sim\posterior]{\Util(\action,\belief,\contract)}}.$$    
    
\end{itemize}
Throughout this section, we denote by $\InBraces{
\widehat{\adist}_T=\InParentheses{
\frac{\sum_{t\in[T]}\ind{\action_t=\action_1}}{T},{\frac{\sum_{t\in[T]}\ind{\action_t=\action_2}}{T}}}
}_{T\geq 0}$ the empirical frequency of a learning agent (see \Cref{sec:learning agent}). 
The convergent result is summarized in \Cref{thm:convergence to Berk-Nash}.

\begin{theorem}\label{thm:convergence to Berk-Nash}
    In environments with two actions, under \Cref{ass:poly contract} and \ref{as:best response bayes}, it is guaranteed that with probability 1, the action frequency of any realized sequence of actions $\InBraces{\action_t}_{t\geq 0}$ by a learning agent converges to a Berk-Nash equilibrium, i.e., $\lim_{t\rightarrow +\infty}\widehat{\adist}_T$ exists and is a Berk-Nash equilibrium.
\end{theorem}

To prove the convergence, we use the asymptotic characterization of the action-frequency $\widehat{\action}_T$ using a differential inclusion by \citet{esponda2021asymptotic}. Differential inclusions extend ordinary differential equations by allowing the derivative of a function to take values from a set, represented as $\dot{x}(t) \in F(x(t))$, rather than being defined by a single-valued function $f(x(t))$.  This approach models systems with uncertain or non-smooth dynamics by admitting multiple possible trajectories at each point. Specifically, in our setting where, for a given action distribution $\adist(t)$, the agent may have various possible posterior beliefs that minimize the KL divergence with respect to $\adist(t)$, each associated with distinct best responses. Additionally, a solution of a differential inclusions is when a point $x$ satisfies the condition $\textbf{0} \in F(x)$, meaning that the function  remains constant or stationary at $x$ as one of the possible behaviors allowed by the inclusion.


\begin{lemma}[\citealp{esponda2021asymptotic}]\label{thm:differential inclusion to action frequency}
Consider the continuous time differential inclusion:
    \begin{align*}
    \dot{\adist}(t) \in \Delta\InParentheses{\BR{\adist(t)}} - \adist(t).
\end{align*}
Moreover, for a starting condition $\adist(0)=\adist\in \Delta(\actions)$, denote by $S_{\adist}^T$ the feasible solution of the differential inclusion for $t\in[0,T]$. Then, asymptotically and with probability $1$, for any $T>0$
\begin{align*}
    \lim_{t\rightarrow+\infty} \inf_{S_{\widehat{\action}_t}^T} \sup_{ t'\in[T]}\|\widehat{\action}_{t+t'} - \adist(t') \| = 0
\end{align*}
\end{lemma}

We show that the differential inclusion introduced in \cref{thm:differential inclusion to action frequency} has a single trajectory for all but a finite set of points. 
This is implied by our following lemma, where we show that the best response to beliefs that are updated according to a distribution over actions is unique except for a finite set of distributions. 
This proof draws upon a technical lemma in the appendix, detailing additional properties of beliefs that minimize the KL-divergence for a predetermined distribution over actions.


\begin{lemma}\label{thm:characterization of BR}
There exists a finite set of break-points $\text{BP}=\InBraces{0=\widehat{\action}^{(1)}<\widehat{\action}^{(2)}<\ldots<\widehat{\action}^{(n)}=1}$ and a corresponding set of actions $\InBraces{\action^{(1)},\action^{(2)},\ldots,\action^{(n-1)}}\in \actions^{|n-1|}$ such that for any mixed-distribution over actions $\adist\in\Delta\InParentheses{\actions}\setminus \{\delta_{\action_1},\delta_{\action_2}\}$:
    \begin{align*}
        \BR{\adist}
        =\begin{cases}
            \action^{(k)} \qquad& \text{If $\adist(\action_1)\in \InParentheses{\widehat{\action}^{( k )},\widehat{\action}^{( k+1)}}$,}\\
            \actions &\text{if $\adist(\action_1)\in \text{BP}$.}
        \end{cases}
    \end{align*}
\end{lemma}



Using this characterization, we can further characterize the convergent points of the differential inclusion in \Cref{thm:differential inclusion convergence}. At a high level, the proof first notices that the best-response function $\BR{\cdot}$ is continuous across its domain, except at a finite number of points. With this observation, coupled with the fact that the action distribution space $\Delta(\actions)$ is one-dimensional, we show that for action distributions $\adist(t)$ where $\BR{\adist(t)}$ is multi-valued, such points are either attractors or, once visited, will not be revisited. Thus the differential inclusion converges either to a multi-valued attractor or to a attracting point of the ordinary differential equation.

\begin{lemma}\label{thm:differential inclusion convergence}
Consider the continuous-time Differential Inclusion:
\begin{align*}
    \dot{\adist}(t) \in \Delta\InParentheses{\BR{\adist(t)}} - \adist(t).
\end{align*}
The differential inclusion converges to:
\begin{itemize}
\item a distribution over actions, denoted as $\adist^*$, such that $\BR{\adist^*}=\actions$; or
\item $\delta_{\action_1}$ (or $\delta_{\action_2}$, respectively) when $\BR{\delta_{\action_1}}=\action_1$ (or $\BR{\delta_{\action_2}}=\action_2$, respectively).
\end{itemize}

\end{lemma}

\begin{proof}

We first show that $\lim_{t\rightarrow+\infty}\adist(t)$ exists. Assume towards contradiction that $\lim_{t\rightarrow+\infty}\adist(t)$ does not exists. 
Then there exists times $t<t'$ that satisfy the following:
\begin{itemize}
    \item $\adist(t) = \adist(t')\notin \text{BP}$ 
    \item For any $t^*\in (t,t')$ either  $ \adist(t)(\action_1) < \adist(t^*)(\action_1)$ or $ \adist(t)(\action_1) >\adist(t^*)(\action_1) $. 
\end{itemize} 

Since $\lim_{t\rightarrow+\infty}\adist(t)$ does not converge to a point, and $\adist(t)$ is continuous in $t$, there must be a set of distribution over actions $L=\InBraces{\adist\in \delta(\actions): |t\in \mathbb{N} : \adist(t) = \adist|=+\infty}$ that we visit infinitely many times. Moreover $|L|=+\infty$.\footnote{Assume towards contradiction that $L<+\infty$. By assumption $|L|\geq 2$ and let $\adist_1,\adist_2\in L$ be two distinct elements of $L$. Then by continuity of $L$, any $\adist\in [\adist_1,\adist_2]$ is also in $L$.} Thus there exist distribution over actions $\adist^*\in L\setminus \text{BP}$. We further assume that time $t$ is some time such that $\adist(t)=\adist^*$ and $t'=\min\InBraces{t^*>t: \adist(t^*)=\adist(t)}$. Thus $\adist(t)\notin \InBraces{\adist(t^*):t\in(t,t''')}$, which by continuity of $\adist(t)$ in $t$ implies that for all $t^*\in (t,t''')$ either $\adist(t^*)(\action_1)>\adist(t)(\action_1)$ or $\adist(t^*)(\action_2)>\adist(t)(\action_2)$, which proves the second item. 

For sufficiently small $\epsilon > 0$, by continuity of $\BR{\cdot}$ in $\adist(t)=\adist(t')\notin \text{BP}$ (see \Cref{thm:characterization of BR}), for any time $\widetilde{t}\in [t,t+\epsilon]\cup [t'-\epsilon,t']$, $\BR{\adist(\widetilde{t})}=\BR{\adist(t)}$ which is either $\delta_{\action_1}$ or $\delta_{\action_2}$. We prove the former case as the latter case follows similarly. Observe that:
$$
\dot{\adist}(\widetilde{t}) = \Delta\InParentheses{\BR{\adist(\widetilde{t})}} - \adist(\widetilde{t}) = \delta_{\action_1} - \adist(\widetilde{t})
.$$
Thus for any $\widetilde{t}\in[t,t+\epsilon]\cup [t'-\epsilon,t']$, $\dot{\adist}(\widetilde{t})(\action_1)>0$ and $\dot{\adist}(\widetilde{t})(\action_2) =-\dot{\adist}(\widetilde{t})(\action_1)< 0$. Consequentially, for times $ t < t+\epsilon < t'-\epsilon < t'$,\footnote{We remind that we can choose $\epsilon>0$ arbitrarily small.} it must be the case that $\adist(t+\epsilon)(\action_1) <\adist(t')(\action_1)=\adist(t)(\action_1)<\adist(t+\epsilon)(\action_1)$, a contradiction to item two above.



Therefore, the limit $\lim_{t\rightarrow+\infty}\adist(t)=\adist^*$ exists and must converge to a solution of the differential inclusion. We now demonstrate that any distribution over actions $\Delta(\actions)$ for which $\adist^*(\action_1) \notin \text{BP}$ cannot be a solution to the differential inclusion. Suppose, for the sake of contradiction, that $\adist(\action_1) \notin \text{BP}$, and yet $\adist^*$ is a solution to the differential inclusion. According to \Cref{thm:characterization of BR}, $\BR{\adist^*}$ must be either $\action_1$ or $\action_2$. However, since $\adist^*$ is a solution to the differential inclusion, we have:
$$
\textbf{0}= \delta_{\action} - \adist^* \Leftrightarrow \adist^* = \delta_{\action}, \text{ where }\action\in\actions, 
$$
leading to a contradiction. This is because $\adist^*$ cannot be one of ${\delta_{\action_1}, \delta_{\action_2}}$, given that ${0,1} \subset \text{BP}$. Thus $\adist^*(\action_1)\in \text{BP}$. Moreover, if $\adist^*\notin \InBraces{\delta_{\action_1},\delta_{\action_2}}$, then $\BR{\adist^*}= \actions$ according to \Cref{thm:characterization of BR}.

We now deal with the case where $\adist^*$ is $\delta_{\action_1}$ or $\delta_{\action_2}$. We only prove the case where the differential inclusion converges to $\delta_{\action_1}$. Since $\delta_{\action_1}$ is a solution of the differential inclusion:
$$
    \textbf{0}\in\delta_{\BR{\delta_{\action_1}}} - \delta_{\action_1},
$$
    which implies the claim since $\BR{\delta_{\action_1}}$ is single-valued.
\end{proof}

We show in the lemma below that all convergent points of \Cref{thm:differential inclusion convergence} are Berk-Nash equilibria.

\begin{lemma}\label{thm:differential inclusion convergence to Berk-Nash}
Let $\belief^{(1)} = \argmin_{\belief\in \beliefs}\KL{\dist_{\action_1}}{\belief_{\action_1}}$, $\belief^{(2)} = \KL{\dist_\action}{\belief_{\action_2}}$, and $\adist^*$ be a distribution over actions. Then 
$\adist^*$ is a Berk-Nash equilibrium in the following cases:
\begin{enumerate}[itemindent=24pt, label = Case (\arabic*):]
    \item $|\supp(\adist^*)|=1$, and if $\adist^*=\delta_{\action_1}$ ($\adist^*=\delta_{\action_2}$ resp.) and $\action_1 \in \argmax_{\action\in\actions} \Util(\action,\belief^{(1)},\contract)$ ($\action_2 \in \argmax_{\action\in\actions} \Util(\action,\belief^{(2)},\contract)$ resp.).
    \item $|\supp(\adist^*)|=2$ and if $\BR{\adist}=\actions$.
\end{enumerate}
\end{lemma}

To complete the proof of \Cref{thm:convergence to Berk-Nash}, we need only show that the continuous differential inclusion in \Cref{thm:differential inclusion to action frequency}, converges to an action frequency corresponding to a Berk-Nash equilibrium. This convergence is assured by \Cref{thm:differential inclusion convergence} and \Cref{thm:differential inclusion convergence to Berk-Nash}, which together conclude that the differential inclusion converges to a Berk-Nash equilibrium.

%% file: content/berk_nash_ppad.tex
\section{Computational Challenges for Berk-Nash Equilibria}
\label{sub:ppad}
When there are three or more actions, 
the long-run action frequency may not converge in general for myopic Bayesian learning agents. 
We provide an example illustrating the ergodic divergence in \cref{apx:ergodic_divergence}. 
Given the ergodic divergence in environments with three or more actions, a natural question is the existence of simple learning dynamics beyond myopic Bayesian learning that converges in action frequency, and hopefully it converges to a Berk-Nash equilibrium with a polynomial convergence rate. 
In this section, we show that for a fixed contract $\contract$, it requires at least quasi-polynomial time to compute even an $\epsilon$-Berk-Nash equilibrium. 
This implies that even there exists a simple learning dynamic that converges to a Berk-Nash equilibrium, 
the rate of convergence cannot be polynomial under the ETH assumption for the PPAD class.

To achieve our goal, we reduce the computation of an approximate Nash Equilibrium in a two-player general-sum game to the computation of an $\epsilon$-Berk-Nash equilibrium for an agent with misspecified beliefs given a linear contract. 

\begin{definition}[General-sum Game]
A general sum game is defined as a pair of payoff matrices $(Y, Z) \in \mathbb{R}_+^{n \times n}$ 
where the payoffs of the row player and the column player are $Y_{i,j}$ and $Z_{i,j}$ respectively given any action profile $(i,j)\in [n]\times[n]$.
\end{definition}

In the reduction, there is a direct correspondence between the agent's actions (and respectively, the agent's beliefs) in the contract design instance and the actions of the row player (and respectively, of the column player) in the general-sum game. 
By carefully designing the contract \(\contract\), we ensure that \(\Util(\action, \belief, \contract)\) closely approximates \(e_{\action}^\top Y e_\belief\) and \(\KL{\dist_{\action}}{\belief_{\action}}\) closely approximates \(e_{\action}^\top Z e_\belief\).\footnote{Here, \(e_\action\) (respectively \(e_\belief\)) denotes the unit vector corresponding to the player's action in the general-sum game that aligns with the action \(\action\) (respectively belief \(\belief\)) in the Berk-Nash instance.}
Given this construction, we observe that the optimality and consistency conditions of the Berk-Nash equilibrium, as detailed in \Cref{sec:definition of Berk-Nash}, ensure that each player in the general-sum game is approximately best responding to the other. 
This equivalence thereby establishes an approximate Nash equilibrium for the original game instance.

The natural input description of the Berk-Nash equilibrium instance is in terms of the bit-complexity of the set of beliefs \(\beliefs\) and the contract \(\contract\). 
However, the Kullback-Leibler divergence \(\KL{\dist_{\action}}{\belief_{\action}}\), a crucial element in our construction, might not be well-defined in terms of bit-complexity under our proposed set of beliefs and chosen contract. Therefore, we use the hardness result established by \citet{Rubinstein16} for computing approximate Nash equilibrium under the Exponential Time Hypothesis for the PPAD class. This hypothesis is stated in \Cref{def:eotl}, with the specific hardness result for general-sum games presented in \Cref{thm:inapproximability}. Our main result is formally stated in \Cref{thm:ppad_hard} and the proof is postponed to \Cref{apx:ppad hard}.

\begin{definition}[End-of-the-Line Problem and PPAD Class]\label{def:eotl}
Let $C_P$ and $C_S$ be circuits, referred to as the predecessor and successor circuits respectively. Both circuits take as input and produce as output a binary string of length $m$. We impose the condition that $C_P(0^m) = 0^m$ but $C_S(0^m) \neq 0^m$. The End-of-the-Line problem is to find a binary string $x$ of length $m$ such that one of the following holds:
\begin{itemize}
    \item $C_P(C_S(x)) \neq x$
    \item $C_S(C_P(x)) \neq x$ and $x \neq 0^m$
\end{itemize}

The Polynomial Parity Arguments on Directed graphs (PPAD) class contains all problems that can be reduced to the End-of-the-Line.
\end{definition}

\begin{hypothesis}[Exponential Time Hypothesis for PPAD  \citep{BabichenkoPR16}]\label{hyp:eth}
The Exponential Time Hypothesis (ETH) for PPAD class conjectures that solving the End-of-the-Line problem for circuits with size $O(n)$ (\Cref{def:eotl}) requires $2^{\tilde{\Omega}(n)}$ time.   
\end{hypothesis}

\begin{proposition}[\citealp{Rubinstein16}]\label{thm:inapproximability}

Under the ETH for PPAD, there exists a constant $\epsilon^*>0$ such that computing an $\epsilon^*$-Nash equilibrium in a two-player $n\times n$ game where the payoffs are normalized in $[0,1]$ requires $n^{\log^{1- o(1)}(n)}$ time.\footnote{While the original assertion requires a different normalization of payoffs, it's crucial to note that an \(\epsilon\)-Nash equilibrium derived from a game without payoffs normalized in $[0,1]$ is still a Nash equilibrium for its normalized counterpart.}
\end{proposition}

\begin{theorem}[Hardness for Equilibrium Computation]
\label{thm:ppad_hard}
Under the ETH for PPAD, there exists a constant $\widetilde{\epsilon}^*>0$ such that computing an $\widetilde{\epsilon}^*$ Berk-Nash equilibrium for a misspecified agent with $n$ actions under contract $\contract$ and beliefs $\beliefs$ with bit complexity $\poly(n)$ requires $n^{\log^{1- o(1)}(n)}$ time.
\end{theorem}

%% file: content/efficiency.tex
\section{Optimal Contracts Under Misspecification}
\label{sec:efficiency}
In the previous section, we observed that for environments with two possible actions, the action frequency converges and it converges to a Berk-Nash equilibrium. 
This provides strong micro-foundation for using Berk-Nash equilibrium as equilibrium concepts for contract design problems when the agent has misspecified beliefs and there are only two possible actions. 
In this section, we will focus on environments with only two possible actions, and provide a polynomial time algorithm for computing the optimal contract. 
Moreover, we will show that when the agent has misspecified beliefs, even a small degree of misspecification can lead to a significant revenue loss for the principal.

\subsection{Polynomial Time Algorithms for Computing Optimal Contracts with Two Actions}
\label{sec:polytime md for two actions}

In this section, we provide a polynomial time algorithm for computing the optimal contract under the generic beliefs assumption (\Cref{ass:poly contract}) when there are two possible actions. For the rest of the section, we denote the set of actions by $\actions=\InBraces{\action_1,\action_2}$. 

In the computational results, a subtle point is that in order to compute the Berk-Nash equilibrium or the optimal contract when there are only two actions, it is inevitable to compute the log-likelihood ratio in KL-divergence (see \cref{sec:definition of Berk-Nash} for definitions). This leads to barriers for exact solutions since it is not feasible to compute a logarithmic number to infinity precision. 
However, as we will illustrate in this section, the impossibility arises solely from the computation of logarithmic number. 
To simplify the exposition, we show that polynomial time algorithm exists if we assume oracle access to the precise values of the logarithmic numbers (\cref{thm:poly_algo_two}). 
To implement this algorithm in practice, it suffice to compute the logarithmic numbers with a precision of $O(\epsilon)$, and the resulting contract will only suffer from an $\epsilon$ loss in revenue under $\epsilon$-Berk Nash equilibrium. 

\begin{theorem}\label{thm:poly_algo_two}
Under \Cref{ass:poly contract}, when there are two possible actions, by assuming oracle access to the precise values of the logarithmic numbers, there exists a polynomial time algorithm for computing a contract and a corresponding Berk-Nash equilibrium that maximizes the principal's expected revenue. 
\end{theorem}

To provide a polynomial time algorithm for \cref{thm:poly_algo_two}, 
we first show that under \Cref{ass:poly contract}, in a Berk-Nash equilibrium, the support of the corresponding posterior is at most two.
The missing proofs in this section will be provided in \cref{apx:poly_algo}.
\begin{lemma}\label{lem:sparce posterior}
    Under \Cref{ass:poly contract}, for any Berk-Nash equilibrium $\alpha$, there exists a posterior belief $\posterior^*\in \Delta(\beliefs)$ that satisfies consistency, optimality and $|\supp(\posterior^*)|\leq 2$.
\end{lemma}
By exploiting the sparsity of the support of the posterior belief in a Berk-Nash equilibrium, 
we can perform an exhaustive search over pair of supports for both distribution over actions and posteriors, 
and compute both the contract and the Berk-Nash with the highest expected revenue for the principal under those support constraints. 
More formally, for a given subset of actions $\actions^*\in \InBraces{\InBraces{\action_1},\InBraces{\action_2}, \InBraces{\action_1,\action_2}}$ and beliefs $\beliefs^*\subseteq \beliefs$ such that $|\beliefs^*|\leq 2$, the program in \Cref{program:original} computes a Berk-Nash equilibrium that maximizes the principal's utility such that $\supp(\adist^*)=\actions^*$ and $\beliefs^*\subseteq \argmin_{\belief\in \beliefs} 
\expect[\action\sim\adist^*]{\KL{\dist_{\action}}{\belief_{\action}}}$. 

\begin{figure}{\linewidth-2cm}
\begin{adjustbox}{minipage=\textwidth-6pt,margin=3pt,bgcolor=black!20}
\textbf{Input:} \(\beliefs^*\subseteq \beliefs, \actions^*\subseteq \actions \)
    \begin{align*}
     &  \max & \expect[\action\sim \adist^*]{\expect[\reward\sim \dist_\action] { \reward -\contract(\reward)} }   \\
     & \text{s.t.} &  \expect[\belief\sim \posterior^*] { \Util(\action^*,\belief,\contract)} & \geq  \expect[\belief\sim \posterior^*] {\Util(\action,\belief,\contract) }                  && \forall \action^* \in \actions^*, \action\in \actions  \\
     &             & \expect[\action\sim\adist^*]{\KL{\dist_{\action}}{\belief^*_{\action}}}  & \leq  \expect[\action\sim\adist^*]{\KL{\dist_{\action}}{\belief_{\action}}}.                 && \forall \belief^* \in \beliefs^*, \belief\in \beliefs   \\
     &             & \contract(\reward) \in \reals_+^{|\rewards|}, & \adist^*\in \Delta(\actions^*), \posterior^*\in \Delta(\beliefs^*). 
    \notshow{
     \\
     & \max  &\sum_{\action\in \actions^*,\reward\in \rewards} \adist^*(\action)\cdot \dist_\action(\reward)\cdot \InParentheses{ \reward -\contract(\reward)}  \\
     & \text{s.t.} &  \sum_{\belief \in \belief^*} \posterior^*(\belief)\sum_{\reward\in\rewards} \belief_{\action^*}(\reward) \InParentheses{ \contract(\reward) -\cost(\action^*)}  \geq &  \sum_{\belief \in \belief^*} \posterior^*(\belief)\sum_{\reward\in\rewards} \belief_\action(\reward) \InParentheses{ \contract(\reward) -\cost(\action)}                 && \forall \action^* \in \actions^*, \action\in \actions  \\
     &             & \sum_{\action\in \actions^*}\adist^*(\action)\cdot {\KL{\dist_{\action}}{\belief^*_{\action}}}  \leq &  \sum_{\action\in \actions^*}\adist^*(\action)\cdot{\KL{\dist_{\action}}{\belief_{\action}}}.                 && \forall \belief^* \in \beliefs^*, \belief\in \beliefs   \\
     &             & \contract(\reward)\geq 0  \forall \reward\in \rewards,& \adist^*\in \Delta(\actions^*), \posterior^*\in \Delta(\beliefs^*). 
     }
    \end{align*}
    \caption{Quadratic program for Finding a Berk-Nash equilibrium that maximized the utility of the principal when the posterior (and actions resp.) of the agent has support $\beliefs^*$ ($\actions^*$ resp.).}
    \label{program:original}
\end{adjustbox}
\end{figure}

Supposing that the optimization program in \Cref{program:original} can be computed in polynomial time, 
\cref{lem:sparce posterior} immediately implies that there exists a polynomial time algorithm for \cref{thm:poly_algo_two} since we only need to compute the optimization program in \Cref{program:original} by $3\times (|\beliefs|^2+|\beliefs|)$ times. 

However, it is important to note that computing an optimal solution for \Cref{program:original} is not trivial since the objective and the feasible set in \Cref{program:original} are non-convex, due to the dependence depend on the contract $\contract$ that is also a variable. 
In fact, \Cref{program:original} corresponds to a quadratic program. 
To overcome this difficulty, we transform it into a family of linear programs (formalized later in \Cref{program:lp for berk-nash}) and show that solving the program in \Cref{program:original} is equivalent to solving a family of linear programs, which can be computed in polynomial time. 
For this purpose, 
we first characterize the feasible distribution over actions that the agent can pick in a Berk-Nash equilibrium when her posterior is supported on set $\beliefs^*\subseteq \beliefs$. 
In particular, we show that the feasible distribution over actions lies in an interval (see \Cref{lem:charactere actions}).

\begin{figure}
\begin{adjustbox}{minipage=\textwidth-6pt,margin=3pt,bgcolor=black!20}
\textbf{Input:} \(\beliefs^*\subseteq \beliefs\)
    \begin{align*} 
     & \action_1^{min}:=  \min \action^*(\action_1)   \\
     & \text{s.t. }  \adist^* \in \Delta(\actions) \text{ and }  \sum_{\action\in \actions}\adist^*(\action)\cdot \InParentheses{{\KL{\dist_{\action}}{\belief^*_{\action}}} -  \KL{\dist_{\action}}{\belief_{\action}}} \leq 0  \qquad  \forall \belief^* \in \beliefs^*, \belief\in \beliefs.    
     \\
     \\
     & \action_1^{max}:=  \max \action^*(\action_1)   \\
     & \text{s.t. }  \adist^* \in \Delta(\actions) \text{ and }  \sum_{\action\in \actions}\adist^*(\action)\cdot \InParentheses{{\KL{\dist_{\action}}{\belief^*_{\action}}} -  \KL{\dist_{\action}}{\belief_{\action}}} \leq 0 \qquad   \forall \belief^* \in \beliefs^*, \belief\in \beliefs.   
    \end{align*}~
    \caption{LP for finding feasible distributions over actions that support a Berk-Nash equilibrium over a set of beliefs $\beliefs^*\subseteq \beliefs$.}
    \label{program:actions}
\end{adjustbox}
\end{figure}

\begin{lemma}\label{lem:charactere actions}
Under \Cref{ass:poly contract}, for any set of beliefs $\beliefs^*\subseteq \beliefs$, the programs in \Cref{program:actions} are linear, and have at most \(O(|\beliefs^*| \cdot |\beliefs|)\) constraints over two variables. Furthermore, for every distribution over actions \(\adist^* \in \Delta(\actions)\), the following inequality holds:
\[
\sum_{\action \in \actions} \adist^*(\action) \left( \KL{\dist_{\action}}{\belief^*_{\action}} - \KL{\dist_{\action}}{\belief_{\action}} \right) \leq 0, \quad \forall \belief^* \in \beliefs^*, \belief \in \beliefs,
\]
if and only if \(\adist^*(\action_1)\) lies within the interval \([\action_1^{\min}, \action_1^{\max}]\).\footnote{$\action_1^{\min}$ and $\action_1^{\max}$ are the values computed in \Cref{program:actions}.}
\end{lemma}

Next we provide a characterization of the contracts such that there exists a posterior $\posterior^*\in \Delta(\beliefs^*)$ that incentivizes the agent to pick any action within set $\actions^*$.
We denote this set of contracts as $\contracts$, i.e.,  \begin{align*}
 \contracts = \InBraces{\contract:\rewards\rightarrow\reals_+: \exists \posterior^*\in \Delta(\beliefs^*), 
 \expect[\belief\sim \posterior^*] { \Util(\action^*,\belief,\contract)} \geq  \expect[\belief\sim \posterior^*] {\Util(\action,\belief,\contract) }                 \forall \action^* \in \actions^*, \action\in \actions }.
 \end{align*}

\begin{lemma}\label{lem:union of convex sets}
 For any non-empty set of actions \( \actions^*\subseteq \actions \) and set of beliefs \(\beliefs^* \subseteq \beliefs\) such that $|\beliefs^*|\leq 2$,
 there exists an algorithm for computing constants $\InBraces{A(\reward),B(\reward), C(\reward), D(\reward)}_{\reward\in \rewards}$ that runs in time $\poly\InParentheses{n,\InBraces{\log\InParentheses{\belief_{\action}(\reward)}}_{\belief\in \beliefs,\action\in\actions,\reward\in\rewards},\InBraces{\log\InParentheses{\cost(\action)}}_{\action\in \actions}}$ 
 such that set $\contracts$ can be decomposed as the union of two convex sets:
 \begin{align*}
 \contracts = &\InBraces{\contract:\rewards\rightarrow\reals_+ : \sum_{\reward\in\rewards}A(\reward)\cdot \contract(\reward)\geq \sum_{\reward\in\rewards}B(\reward)\cdot \contract(\reward) \geq 0 } \\
 &\qquad\qquad \bigcup
 \InBraces{\contract:\rewards\rightarrow\reals_+ : \sum_{\reward\in\rewards}C(\reward)\cdot \contract(\reward)\geq \sum_{\reward\in\rewards}D(\reward)\cdot \contract(\reward) \geq 0 }.
 \end{align*}
\end{lemma}

\begin{figure}{\linewidth-2cm}
\begin{adjustbox}{minipage=\textwidth-6pt,margin=3pt,bgcolor=black!20}
\textbf{Input:}
\begin{itemize}
    \item $\beliefs^*\subseteq \beliefs, \actions^*\subseteq \actions$.
    \item convex sets \(\contracts_1,\contracts_2\) such that \(\contracts = \contracts_1\cup\contracts_2\) (promised by \Cref{lem:union of convex sets}).
    \item constants $\action_1^{min},\action_1^{max}$ (promised by \Cref{lem:charactere actions}).
\end{itemize} 

Consider the following LP parametrized by $\action_1^{val}\in [0,1]$ and $\widehat{\contracts}\subseteq \contracts$ :
    \begin{align*} \textsc{VAL}\InParentheses{\action_1^{val},\widehat{\contracts}} = \max  &\quad \expect[\action\sim \adist^*]{\expect[\reward\sim \dist_\action] { \reward -\contract(\reward)} }\\
       & \quad:= \adist^*(\action_1)\sum_{\reward\in\rewards}\dist_{\action_1}(\reward) \cdot \InParentheses{\reward - \contract(\reward)} + \adist^*(\action_2)\sum_{\reward\in\rewards}\dist_{\action_2}(\reward) \cdot \InParentheses{\reward - \contract(\reward)}   \\
      \text{s.t.} &\quad  \contract \in \widehat{\contracts} , \adist^*(\action_1) = \action_1^{val}, \text{ and } \adist^*(\action_2) = 1- \action_1^{val}.
    \end{align*}
    Return the contract $\contract^*$ that maximizes the value of the parametrized LP $\textsc{VAL}\InParentheses{\action_1^{val}, \widehat{\contracts}}$ for $\action_1^{val}\in \InBraces{\action_1^{min},\action_1^{max}}$ and $\widehat{\contracts}\in \InBraces{\contracts_1,\contracts_2}$. Formally let: $$(\action_1^*,\contracts^*)\in \argmax_{\action_1^{val}\in \InBraces{\action_1^{min},\action_1^{max}},\widehat{\contracts}\in \InBraces{\contracts_1,\contracts_2}}\textsc{VAL}\InParentheses{\action_1^{val}, \widehat{\contracts}}.$$

    Return the $\contract^*$ that maximizes the value of $\textsc{VAL}\InParentheses{\action_1^*, \contracts^*}$
    \caption{LP for computing a Berk-Nash equilibrium that maximized the utility of the principal when the posterior (and actions resp.) of the agent has support $\beliefs^*$ ($\actions^*$ resp.).}
    \label{program:lp for berk-nash}
\end{adjustbox}
\end{figure}

Now we are ready to formally show how to compute a contract that maximizes the revenue. We formally show the LP that computes the contract that induces the highest revenue for the principal over Berk-Nash equilibriums where the agent's actions are supported on $\actions^*$ and her posterior is supported on $\beliefs^*$. In a high-level idea, when we search over contracts for Berk-Nash equilibriums where the agent's action profile and posterior are supported on fixed sets $\actions^*$ and $\beliefs^*$, we can overcome the non-convexity of the objective in \Cref{program:original} by only considering the two extreme distributions over actions as indicated by \cref{lem:charactere actions}. To further overcome the non-convexity of possible set of contracts $\contracts$ (as defined in \Cref{lem:union of convex sets}), we decompose set $\contracts$ into the union of two convex sets and search in each convex set independently. We formally prove this intuition in \Cref{thm:lp for berk-nash}

\begin{lemma}\label{thm:lp for berk-nash}
Under \Cref{ass:poly contract}, for any non-empty set of actions \( \actions^*\subseteq \actions \) and any set of beliefs \(\beliefs^* \subseteq \beliefs\) such that \(|\beliefs^*|\leq 2\), 
there exists a polynomial time algorithm that computes the optimal solution for the optimization program in 
\Cref{program:lp for berk-nash}, 
which outputs the same contract as in \Cref{program:original}. 
\end{lemma}

%% file: content/lower_bound.tex
\subsection{Revenue Losses from Misspecification}\label{sec:unhappy principal}

In many practical application, it is often the case that the agent may suffer from at least a small degree of misspecification.
In this section, we show that the principal's expected revenue in a Berk-Nash equilibrium can deviate significantly from the optimal revenue achievable in a correctly specified setting even when there is not a small degree of misspecification. 

\begin{definition}[$\epsilon$-misspecification]
For any $\epsilon\geq 0$, an agent has $\epsilon$-misspecified beliefs if there exists a belief $\belief \in \beliefs$ such that for any action $a\in A$, we have $\KL{\dist_a}{\belief_{a}}\leq \epsilon$.
\end{definition}

We say an agent has correctly specified beliefs if he has $\epsilon$-misspecified beliefs for $\epsilon = 0$. 
Note that in our definition of $\epsilon$-misspecification, we make a strong requirement that a single belief is close to the true underlying data generating process for all possible actions.\footnote{An alternative weaker version of the definition would only require the existence of a belief with small KL-divergence for each possible action. } 
Our result shows that even under such strong requirement of close to correct specification, 
the loss in expected revenue can still be unbounded even when $\epsilon$ approaches $0$. 

\begin{theorem}\label{thm:rev_loss}
For any $\epsilon > 0$, there exists an instance with $\epsilon$-misspecified beliefs such that the multiplicative gap between the optimal revenue from the correctly specified model and optimal revenue from the misspecified model is at least $1.81$. 
\end{theorem}

The proof of \cref{thm:rev_loss} relies on the construction in the following instance.

\begin{instance}\label{instance:unhappy principal}
We consider a contract design problem characterized by parameters $p \in \InBrackets{3/4,1} $, $c\in\InBrackets{0,1/2}$, $\delta \in (0, 1-p)$. There are three potential rewards $\rewards=\InBraces{0,1,2}$.  
The action space is denoted as $\actions=\InBraces{ \action_g, \action_{b}}$ and the agent has cost $c$ for choosing action $\action_g$ (e.g., $\cost(\action)=c\cdot\mathds{1}[\action=\action_g]$). 

\noindent\textbf{True reward distribution:} The distribution $\dist$ over rewards for each action is parametrized by $p$ as follows:
\begin{align*}
    \Pr_{\reward\sim \dist_{\action_g}}[\reward]=
    \begin{cases}
        \delta  &\text{if } \reward=2,\\
        p  &\text{if } \reward=1, \\
        1-p-\delta  &\text{if } \reward= 0.
    \end{cases}\qquad\qquad
    \Pr_{\reward\sim \dist_{\action_b}}[\reward]=
    \begin{cases}
    0 &\text{if } \reward=2,\\
    1-p &\text{if } \reward= 1,\\
    p &\text{if } \reward=0.    
    \end{cases}
\end{align*}

\noindent\textbf{Misspecified beliefs:} 
The agent has a singleton belief $B$ specified as follows:
\begin{align*}
    \Pr_{\reward\sim \belief'_{\action_g}}[\reward]=
    \begin{cases}
        \delta &\text{if } \reward=2,\\
        p &\text{if } \reward=1, \\
        1-p-\delta &\text{if } \reward= 0.
    \end{cases}\qquad\qquad
    \Pr_{\reward\sim \belief'_{\action_b}}[\reward=1]=
    \begin{cases}
    \delta &\text{if } \reward=2,\\
    1-p &\text{if } \reward= 1,\\
    p-\delta &\text{if } \reward=0.    
    \end{cases}
\end{align*}
\end{instance}

It is easy to verify that the instance constructed in \cref{instance:unhappy principal} satisfies $\epsilon$-misspecified beliefs for $\delta = \frac{\epsilon}{2}$. 
This is formally stated in \cref{lem: KL for unhappy}. 
The proof of the lemmas in the remaining section are provided in \cref{apx:revenue_loss}.
\begin{lemma}\label{lem: KL for unhappy}
In \Cref{instance:unhappy principal}, for any action $\action\in \InBraces{\action_g,\action_b}$, 
    $\KL{\dist_a}{\belief'_{a}}\leq 2\delta$. 
\end{lemma}

Moreover, when the agent is correctly specified, we show that there exists a contract such that the expected revenue is $p-c$. Intuitively, in correctly specified model, by offering a contract with positive reward $\frac{c}{\delta}$ only for reward $2$ (the outcome that maximizes the likelihood ratio), the agent is incentivized to choose action $a_g$ and the principal extract the full surplus from the agent. 
\begin{lemma}\label{lem:correct specification}
In \Cref{instance:unhappy principal}, when the agent has a correctly specified belief, there exists a contract that attains the expected revenue of $p+2\delta -c$ for the principal.
\end{lemma}

Finally, when the agent has misspecified beliefs, i.e., $\epsilon>0$, the optimal revenue of the principal under Berk-Nash equilibrium is at most $1-p$. The main intuition is that even with small misspecification in beliefs, the likelihood ratio for outcome $2$ reduces from infinite to 1. 
Under misspecified belief $B$, the maximum likelihood ratio reduces to $\frac{p}{1-p}$. To incentivize the agent to choose action $a_g$, the principal has to leave significant utility for the agent, which in turn reduces the expected revenue. 
\begin{lemma}\label{thm:unhappy principal}
In \Cref{instance:unhappy principal}, when the agent has misspecified beliefs, the optimal revenue given any contract is at most 
$\max\{1-p, p+2\delta -\frac{c\cdot p}{2p-1}\}$.
\end{lemma}

Finally, \cref{thm:rev_loss} holds by combining \cref{lem: KL for unhappy,lem:correct specification,thm:unhappy principal} for \cref{instance:unhappy principal} with parameters $\delta\to 0$, $p=0.86$ and $c=0.6$, 
since the ratio between the principal's utility in the correctly specified setting over the misspecified setting is:
\begin{align*}
    \lim_{\epsilon\rightarrow0}\frac{p+2\delta - c}{ 
    \max\left( 
    1-p,p+2\delta-\frac{c\cdot p}{2p-1}
    \right)}\geq1.81.
\end{align*}

%% file: content/conclusion.tex
\section{Conclusions and Discussions}
\label{sec:conclude}
In our paper, we consider the contract design problems for an agent with misspecified beliefs. We show that when there are two possible actions, the action frequencies of a myopic Bayesian learning agent converge, eventually reaching a Berk-Nash equilibrium. Subsequently, we present a polynomial-time algorithm for computing the optimal contract of the principal that maximizes her expected revenue under Berk-Nash equilibria. However, in cases with three or more actions, we show that the action frequencies may fail to converge, and it requires at least quasi-polynomial time to compute even an approximate Berk-Nash equilibrium. This poses an intriguing open question regarding the identification of appropriate equilibrium concepts for contract design with three or more actions when the agent has misspecified beliefs. Lastly, our paper highlights that the principal may incur a significant revenue loss even with minor degrees of misspecification, underscoring the importance of studying robust contract design under such circumstances.

%% file: appendix/eqn_learning.tex
\section{Missing Proofs for Convergence with Binary Actions}
\label{subapx:convergence_binary}

\begin{lemma}\label{lem:best response bayes}
    For a fixed contract $\contract$, let $\beliefs'=\InBraces{ \belief\in \beliefs : \Util(\action_1,\belief,\contract) \neq \Util(\action_2,\belief,\contract)}$ be the set of beliefs where an agent strictly prefers one of the two actions and $\prior$ be any prior of the agent supported on $\beliefs$. We consider the prior $\prior'$ supported on $\beliefs'$ such that:
    $$
    \prior'(\belief)\propto \prior(\belief)\ind{\belief\in \beliefs'} \qquad\qquad \forall \belief\in \beliefs' .$$
    \sloppy For a realized sequence of actions that the agent took and outcomes she observed $\action^{(1)},\reward^{(1)},\ldots, \action^{(T)},\reward^{(T)}$, let $\posterior$ ($\posterior'$ resp.) be the potesterior of the agent from prior $\prior$ ($\prior'$ resp.). Then:
    \begin{align*}
\argmax_{\action\in\actions} \expect[\belief\sim\posterior]{\Util(\action,\belief,\contract)} = \argmax_{\action\in\actions} \expect[\belief\sim\posterior']{\Util(\action,\belief,\contract)}.
\end{align*}
\end{lemma}

\begin{proof}
    It suffices to prove that $\expect[\belief\sim\posterior]{\Util(\action_1,\belief,\contract)} - \expect[\belief\sim\posterior]{\Util(\action_2,\belief,\contract)}\propto\expect[\belief\sim\posterior']{\Util(\action_1,\belief,\contract)} - \expect[\belief\sim\posterior']{\Util(\action_2,\belief,\contract)}$.
    By Bayes rule, realize the following equalities:
    \begin{align*}
        & \expect[\belief\sim\posterior]{\Util(\action_1,\belief,\contract)} - \expect[\belief\sim\posterior]{\Util(\action_2,\belief,\contract)}\\ = & \sum_{\belief\in \beliefs} \frac{\prior(\belief)\cdot \prod_{t=1}^T \belief_{\action^{(t)}}(\reward^{(t)} ) }{\sum_{\belief'\in \beliefs} \prior(\belief')\cdot \prod_{t=1}^T \belief'_{\action^{(t)}}(\reward^{(t)} )}\cdot\InParentheses{\Util(\action_1,\belief,\contract) - \Util(\action_2,\belief,\contract)} \\
        = & \frac{1}{\sum_{\belief'\in \beliefs} \prior(\belief')\cdot \prod_{t=1}^T \belief'_{\action^{(t)}}(\reward^{(t)})} \sum_{\belief\in \beliefs'} \InParentheses{\prior(\belief)\cdot \prod_{t=1}^T \belief_{\action^{(t)}}(\reward^{(t)} ) }\cdot\InParentheses{\Util(\action_1,\belief,\contract) - \Util(\action_2,\belief,\contract)}.
    \end{align*}

    In the first equality, we simply used the Bayes rule. In the second equality we used the fact that $\Util(\action_1,\belief,\contract) = \Util(\action_2,\belief,\contract)$ for $\belief \notin \beliefs'$. Similarly, we have that:

        \begin{align*}
        & \expect[\belief\sim\posterior']{\Util(\action_1,\belief,\contract)} - \expect[\belief\sim\posterior']{\Util(\action_2,\belief,\contract)}\\ = & \sum_{\belief\in \beliefs'} \frac{\prior'(\belief)\cdot \prod_{t=1}^T \belief_{\action^{(t)}}(\reward^{(t)} ) }{\sum_{\belief'\in \beliefs} \prior'(\belief')\cdot \prod_{t=1}^T \belief'_{\action^{(t)}}(\reward^{(t)} )}\cdot\InParentheses{\Util(\action_1,\belief,\contract) - \Util(\action_2,\belief,\contract)} \\
        = & \sum_{\belief\in \beliefs'} \frac{\frac{\prior(\belief)}{\sum_{\belief''\in\beliefs'}\prior(\belief'')}\cdot \prod_{t=1}^T \belief_{\action^{(t)}}(\reward^{(t)} ) }{\sum_{\belief'\in \beliefs} \frac{\prior(\belief)}{\sum_{\belief''\in\beliefs'}\prior(\belief'')}\cdot \prod_{t=1}^T \belief'_{\action^{(t)}}(\reward^{(t)} )}\cdot\InParentheses{\Util(\action_1,\belief,\contract) - \Util(\action_2,\belief,\contract)} \\
        = & \frac{1}{\sum_{\belief'\in \beliefs'} \prior(\belief')\cdot \prod_{t=1}^T \belief'_{\action^{(t)}}(\reward^{(t)})} \sum_{\belief\in \beliefs'} \InParentheses{\prior(\belief)\cdot \prod_{t=1}^T \belief_{\action^{(t)}}(\reward^{(t)} ) }\cdot\InParentheses{\Util(\action_1,\belief,\contract) - \Util(\action_2,\belief,\contract)} \\
        = & \frac{\sum_{\belief'\in \beliefs} \prior(\belief')\cdot \prod_{t=1}^T \belief'_{\action^{(t)}}(\reward^{(t)})}{\sum_{\belief'\in \beliefs'} \prior(\belief')\cdot \prod_{t=1}^T \belief'_{\action^{(t)}}(\reward^{(t)})}\InParentheses{ \expect[\belief\sim\posterior]{\Util(\action_1,\belief,\contract)} - \expect[\belief\sim\posterior]{\Util(\action_2,\belief,\contract)}} ,       
    \end{align*}

    where the first equality follows because for $\belief\notin\beliefs'$, $\prior'(\belief)=0$. The second equality follows by definition of $\prior'$. The last equality concludes the proof.
\end{proof}

\begin{prevproof}{thm:differential inclusion to action frequency}
    The proof is an instantiation of Theorem~2 by \citet{esponda2021asymptotic} by noticing that the best-response correspondence is upper hemicontinous, and Assumptions~1-2 in \citet{esponda2021asymptotic} are satisfied since the belief space $\beliefs$ is finite, and the prior has full support by assumption.
\end{prevproof}

\notshow{
\begin{observation}
    For any distribution over actions $\adist^*\in \Delta(\actions)\setminus \delta_{\action_1},\delta_{\action_2}$, a necessery and sufficient condition so that there exists a posterior $\posterior^*\in \Delta(\beliefs)$ such that $(\adist^*,\posterior^*)$ is a Berk-Nash is that $\BR{\adist}=\actions$.
\end{observation}
\begin{proof}
    We first prove that it is necessery. Assume towards contradictions that $\BR{\adist^*}=\InBraces{\action_1}$ (the case where $\BR{\adist^*}=\InBraces{\action_2}$ follows similarly). Then by definition of $\BR{\adist^*}$, the optimality condition can never be satisfied for action $\action_2$.

    Now we prove that is sufficient. By definition of $\BR{\adist^*}$, there exists posteriors over beliefs $\posterior_1,\posterior_2\in \Delta(\beliefs{\adist^*})$ such that:
    \begin{align*}
        \action_1 \in \argmax_{\action\in\actions} \expect[\belief\sim\posterior_1]{\Util(\action,\belief,\contract)},\\
        \action_2 \in \argmax_{\action\in\actions} \expect[\belief\sim\posterior_2]{\Util(\action,\belief,\contract)}.
    \end{align*}
    Thus, there exists a posterior $\posterior'\in \Delta(\beliefs\InParentheses{\adist^*})$ such that: 
        \begin{align*}
        \action_1,\action_2 \in \argmax_{\action\in\actions} \expect[\belief\sim\posterior']{\Util(\action,\belief,\contract)},
    \end{align*}
    which implies that the optimality condition for the Berk-Nash $(\adist^*,\posterior^*)$ is satisfied. The consistensy condition is satisfied by definition of $\beliefs(\adist^*)$.
\end{proof}
}

\begin{prevproof}{thm:characterization of BR}
Our key technical lemma for the proof lies in \Cref{lem:char of inclusion} below.
We introduce some additional notation:
\begin{itemize}

    \item For any belief $\belief\in \beliefs$, we denote the set of action-profiles that minimize the KL divergence of $\belief$ by:
    $$\Delta(\actions)_{\belief}=\InBraces{\adist\in \Delta(\actions): \belief\in \argmin_{\belief\in\beliefs}\expect[\action\sim\adist]{\KL{\dist_{\action}}{\belief_{\action}}}}.$$

    \item For any belief $\belief\in \beliefs$, we denote the best-response of the player according to $\belief$ by:
    $$\BR{\belief}=\argmax_{\action\in\actions} {\Util(\action,\belief,\contract)}.\footnote{\text{By \Cref{as:best response bayes}, $\BR{\belief}$ contains only one element.} }$$

\end{itemize}

\begin{lemma}\label{lem:char of inclusion}
    Consider a misspecified learning instance satisfying \Cref{ass:poly contract} and \Cref{as:best response bayes}.
\begin{itemize}
    \item For any distribution over actions $\adist^*\in \Delta(\actions)$, $|\beliefs\InParentheses{\adist^*}|\leq 2$.
    
    \item For any belief $\belief\in \beliefs$, the set $\Delta(\actions)_{\belief}$ is compact and convex. 

    \item For any pair of beliefs $\belief,\belief'\in \beliefs$, there exists at most one distribution over actions $\adist^*\in \Delta(\actions)$ such that  $\adist^*\in \Delta(\actions)_{\belief} \cap \Delta(\actions)_{\belief'}$. 
    
    \item For any distribution over actions $\adist$, $\BR{\adist} = \cup_{\belief\in \beliefs(\adist)}\BR{\belief}$.
    
\end{itemize}    
\end{lemma}

\begin{proof}
We prove the first item of the statement. Assume towards contradiction that there exists three distinct beliefs $\belief_1,\belief_2,\belief_3\in \beliefs$ such that:
$$\belief_1,\belief_2,\belief_3\in \argmin_{\belief\in \beliefs} 
\expect[\action\sim\adist^*]{\KL{\dist_{\action}}{\belief_{\action}}}.$$
Then the following system of linear equations is feasible:
\notshow{
\begin{align*}
    \left\langle \adist^*, \kl(\belief_1)\right\rangle = & \min_{\belief\in \beliefs} 
\expect[\action\sim\adist^*]{\KL{\dist_{\action}}{\belief_{\action}}},\\
\left\langle \adist^*, \kl(\belief_2)\right\rangle = & \min_{\belief\in \beliefs} 
\expect[\action\sim\adist^*]{\KL{\dist_{\action}}{\belief_{\action}}},\\
\left\langle \adist^*, \kl(\belief_3)\right\rangle = & \min_{\belief\in \beliefs} 
\expect[\action\sim\adist^*]{\KL{\dist_{\action}}{\belief_{\action}}},
\end{align*}
which imply that the following system of linear equations is true:
}
\begin{align*}    
 \left\langle \adist^*, \kl(\belief_1)- \kl(\belief_3)\right\rangle = & 0\\
\left\langle \adist^*, \kl(\belief_2)- \kl(\belief_3)\right\rangle = & 0,
\end{align*}
a contradiction since $\adist^*\in \Delta(\actions)\subseteq \reals^2$ and vectors  $\kl(\belief_1)- \kl(\belief_3)$ and $\kl(\belief_2)- \kl(\belief_3)$ are linear independent (\Cref{ass:poly contract}).\footnote{Observe that the only valid solution to the system is $\adist^*=(0,0)$.} Hence for any distribution over actions $\adist^*$,  $|\beliefs(\adist^*)|\leq 2$.

The second item of the statement follows since $\Delta(\actions)_{\belief}$ is the feasible set to the following system of linear inequalities:
\begin{align*}
    \left\langle \adist^*, \kl(\belief')\right\rangle \geq \left\langle \adist^*, \kl(\belief)\right\rangle,\adist^*\in \Delta(\actions) \qquad\qquad \text{for all }\belief'\in \beliefs,
\end{align*}
and as such the set of feasible solutions is convex and compact.

For the third item of the statement, for any pair of beliefs $\belief,\belief'\in \beliefs$, and distribution over actions $\adist^*\in \Delta(\actions)$ such that $\beliefs(\adist^*)=\InBraces{\belief,\belief'}$ must satisfy the following system of linear equation:
\begin{align*}
    &\adist^*(\action_1)\KL{\dist_{\action_1}}{\belief_{\action_1}} + (1- \adist^*(\action_1))\KL{\dist_{\action_2}}{\belief_{\action_2}} \\
    =& \adist^*(\action_1)\KL{\dist_{\action_1}}{\belief'_{\action_1}} + (1- \adist^*(\action_1))\KL{\dist_{\action_2}}{\belief'_{\action_2}}.
\end{align*}
By \Cref{ass:poly contract}, $\KL{\dist_{\action_2}}{\belief'_{\action_2}} \neq \KL{\dist_{\action_2}}{\belief_{\action_2}}$ and the system has at most one solution:
\begin{align*}
\adist^*(\action_1) = \frac{\KL{\dist_{\action_2}}{\belief'_{\action_2}} - \KL{\dist_{\action_2}}{\belief_{\action_2}}}{\KL{\dist_{\action_1}}{\belief_{\action_1}} -\KL{\dist_{\action_1}}{\belief_{\action_2}} - \KL{\dist_{\action_1}}{\belief'_{\action_1}}  + \KL{\dist_{\action_1}}{\belief'_{\action_2}}   }.
\end{align*}

For the last item of the statement, we take cases depending on the size of $|\beliefs(\adist)|$ (we know from item one that $|\beliefs(\adist)|\leq 2$). 

For the case $|\beliefs(\adist)|=1$, we denote the unique belief in $\beliefs(\adist)$ by $\belief$. The proof follows by the fact that $\Delta(\beliefs(\adist))= \delta_{\belief}$:
$$\BR{\adist}=\InBraces{\action^*\in \actions: \exists \posterior\in \Delta(\beliefs(\adist^*)): \action^* \in \argmax_{\action\in\actions} \expect[\belief\sim\posterior]{\Util(\action,\belief,\contract)}}=\BR{\belief}.$$

For the case where $|\beliefs\InParentheses{\adist}|=2$, we denote the two beliefs by $\belief,\belief'$. We first take the subcase where $\BR{\belief}=\BR{\belief'}$ and denote the unique maximizing action by $\action$ and the other action by $\action'$. By definition, for any distribution over beliefs in $\beliefs\InParentheses{\adist}$, and $\posterior\in \Delta(\belief,\belief')$: 

$$\expect[\belief\sim\posterior]{\Util(\action_{\adist},\belief,\contract)} >\expect[\belief\sim\posterior]{\Util(\action',\belief,\contract)},$$
which concludes this subcase.
The subcase where $\BR{\belief}\neq \BR{\belief'}$ follows by taking the posteriors $\posterior=\delta_{\belief}$ (posterior $\posterior'=\delta_{\belief'}$ resp.), and noting that $\argmax_{\action\in\actions} \expect[\belief\sim\posterior]{\Util(\action,\belief,\contract)}=\BR{\belief}$ ($\argmax_{\action\in\actions} \expect[\belief\sim\posterior']{\Util(\action,\belief,\contract)}=\BR{\belief'}$ resp.).
\end{proof}

We proceed with the proof of \Cref{thm:characterization of BR}. We first show that the set $\text{JP}=\{\adist\in \Delta(\actions): \BR{\adist}=\actions\}$ is finite. We then show that for any $\adist\in \text{JP}$, two unique beliefs $\belief, \belief' \in \beliefs$ can be precisely identified, such that $\beliefs(\adist)=\{\belief, \belief'\}$. Following item three of \Cref{lem:char of inclusion}, this suggests that for a given pair of beliefs $\InBraces{\belief, \belief'}$, a unique action distribution $\adist$ exists, such that $\beliefs(\adist)=\InBraces{\belief, \belief'}$. Since set $\beliefs$ is finite, the finiteness of $\text{JP}$ is consequently inferred. Assume towards contradiction that $|\beliefs(\adist)|= 1$, then $\BR{\adist}=\BR{\beliefs(\adist)}$ which cannot be $\actions$, as contradicted by \Cref{as:best response bayes}. Thus, according to item one of \Cref{lem:char of inclusion}, it's must be that $|\beliefs(\adist)|= 2$.

Subsequently, we define $\text{BP} = \{\adist(\action_1):\adist\in \text{JP}\}\cup \InBraces{0,1}$ and denote its elements as follows:
$$\text{BP}=\InBraces{0=\widehat{\action}^{(1)}<\widehat{\action}^{(2)}<\ldots<\widehat{\action}^{(n)}=1}.$$
Given the definition of $\text{JP}$, for each $k\in[n-1]$, there is a belief $\belief$ such that for all $\adist\in \Delta(\actions)$ with $\adist(\action_1)\in (\widehat{\action}^{(k)},\widehat{\action}^{(k+1)})$, it holds that $|\beliefs(\adist)|=\InBraces{\belief}$, implying that $\BR{\adist}=\BR{\belief}$. By \Cref{as:best response bayes}, this results in $\BR{\adist}$ being either $\action_1$ or $\action_2$.
\end{prevproof}

\begin{prevproof}{thm:differential inclusion convergence to Berk-Nash}
We prove the first part of the statement for the case where $\adist^*=\delta_{\action_1}$, since the case $\adist^*=\delta_{\action_2}$ follows similarly. When $\action_1 \in \argmax_{\action\in\actions} \Util(\action,\belief^{(1)},\contract)$, the distribution over actions $\delta_{\action_1}$ coupled with the posterior $\delta_{\belief^{(1)}}$ satisfies both the consistensy and optimality condition (see \Cref{sec:definition of Berk-Nash}).

Now we prove the second part of the statement. Observe that If $\BR{\adist}=\actions$, then there exists posterior $\posterior\in \Delta\InParentheses{\beliefs\InParentheses{\adist}}$ ($\posterior'\in \Delta\InParentheses{\beliefs\InParentheses{\adist}}$ resp.) such that $\action_1=\argmax_{\action\in\actions} \expect[\belief\sim\posterior]{\Util(\action,\belief,\contract)}$ ($\action_2=\argmax_{\action\in\actions} \expect[\belief\sim\posterior']{\Util(\action,\belief,\contract)}$ resp.). Thus, there exists $\posterior^*\in \Delta(\belief,\belief')$ such that:
$$
\expect[\belief''\sim \posterior^*]\Util(\action_1,\belief'',\contract) =\expect[\belief''\sim \posterior^*]\Util(\action_2,\belief'',\contract).
$$
The optimality condition with respect to posterior $\posterior^*$ is satisfied by definition, and the consistensy condition is satisfied since $\InBraces{\belief,\belief'}\in \beliefs\InParentheses{\adist^*}$.
\end{prevproof}

\section{Missing Proofs for Computational Hardness}\label{apx:ppad hard}
\begin{prevproof}{thm:ppad_hard}
Consider a two-player general-sum game $(Y,Z)\in \reals_+^{n}\times \in \reals_+^{n}$. We show that if we can compute a $\widetilde{\epsilon}^*$ Berk-Nash equilibrium in $o(n^{\log^{1- o(1)}(n)})$ time, then we can compute an $\epsilon^*$-Nash equilibrium (see \Cref{thm:inapproximability} for what $\epsilon^*$ is) in $o(n^{\log^{1- o(n)}(n)})$ time, a contradiction of \Cref{thm:inapproximability} under \Cref{hyp:eth}. The following techical lemma is the building block for our construction.

    \begin{lemma}\label{lem:red game to contract}
Let $\widetilde{Y}=\begin{pmatrix}\widetilde{Y}_{1,1}&\ldots & \widetilde{Y}_{1,n}\\ &\vdots \\ \widetilde{Y}_{n,1}&\ldots & \widetilde{Y}_{n,n}\end{pmatrix} \in \reals_+^{n\times n}$ and $\widetilde{Z}=\begin{pmatrix}\widetilde{Z}_{1,1}&\ldots & \widetilde{Z}_{1,n}\\ &\vdots \\ \widetilde{Z}_{n,1}&\ldots & \widetilde{Z}_{n,n}\end{pmatrix}\in \Nat^{n\times n}$ be given matrices and consider a precision parameter $\epsilon' \geq 0$. 

We can construct a misspecified contract design setting wherein an agent has a set of $n$ actions, denoted as $\InBraces{\action_1,\ldots,\action_n}$, where action $\action_i$ produces a distribution $F_{\action_i}$ over $(n+1)\cdot n$ rewards, denoted by $\rewards$. The agent has $n$ beliefs, denoted as $\InBraces{\belief_1,\ldots,\belief_n}$, and is subject to a contract $\contract:\rewards\rightarrow \mathbb{R}$. For a given belief $\belief_i$, and action $\action_j$ we have:
\begin{align*}
\InAbsolute{V(\action_i,\belief_j,\contract) - \widetilde{Y}_{i,j}} \leq \epsilon',\\
\InAbsolute{\KL{F_{\action_i}}{\belief_{(j,\action_i)}} - \widetilde{Z}_{i,j}}\leq \epsilon'.
\end{align*}
Moreover the bit complexity of the contract $\contract$ and the set of beliefs is,
$$\poly\InParentheses{n,\max_{i\in[n],j\in[m]}\log\InParentheses{\widetilde{Y}_{i,j}}, \max_{i\in[n],j\in[m]}\widetilde{Z}_{i,j}, \log\InParentheses{\frac{1}{\epsilon'}} }.$$
\end{lemma}

\begin{proof}
In the contract design instance, we will construct $n+1$ distinct outcomes for each chosen action. In particular, we denote the $n+1$ possible rewards induced by action $\action_i$ as $\rewards_{\action_i}=\InBraces{\reward^{(i)}_0,\reward^{(i)}_1,\ldots,\reward^{(i)}_n}$. 
We will construct the instance such that for both the underlying distribution and the subject beliefs of the agent, only outcomes within $\rewards_{\action_i}$ will occur with positive probability when the chosen action is $\action_i$.

Our construction is parametrized by integer $k\geq \max\InParentheses{4,\log_2(4\cdot n)}$. 
Let $\widetilde{e} = \left\lceil\frac{e}{2^{-k}}\right\rceil\cdot 2^{-k}$ be the Euler constant truncated at the $k$-th bit. Observe the following inequalities:
\begin{align}\label{eq:bound on tilde first}
    e \leq \widetilde{e} \leq e + 2^{-k} \overset{(\log(e+x)\leq 1+\frac{x}{e})}{\Rightarrow}& 1 \leq \log (\widetilde{e}) \leq \log(e + 2^{-k})\leq 1 + \frac{2^{-k}}{e}, \\
     \log_{\widetilde{e}}(4\cdot n) \leq \log_{2}(4\cdot n) \leq  k \Rightarrow& n\widetilde{e}^{-k}\leq \frac{1}{4},\\
     \widetilde{e}\geq e, \widetilde{Z}_{i,j}\geq 1\Rightarrow & \widetilde{e}^{-\widetilde{Z}_{i,j}}\leq \frac{1}{2}.\label{eq:bound on tilde last}
\end{align}
We first construct the underlying distribution over rewards induced by action $\action_i$:
\begin{align*}
    F_{a_i}(\reward) =
    \begin{cases}
        1-n\cdot \widetilde{e}^{-k} \qquad &\text{if } \reward=\reward^{(i)}_0,\\
        \widetilde{e}^{-k} \qquad & \text{if } r \in R_{\action_i}\backslash\{\reward^{(i)}_0\},\\
        0 &\text{otherwise}.
    \end{cases}
\end{align*}
Moreover, under belief $\belief_{(j,\action_i)}$ the agent has the following distribution over rewards:
\begin{align*}
    \belief_{(j,\action_i)}(\reward) =
    \begin{cases}
         \widetilde{e}^{-\widetilde{Z}_{i,j}} \qquad &\text{if } \reward=\reward^{(i)}_0,\\
         1-(n-1)\cdot \widetilde{e}^{-k}-\widetilde{e}^{-\widetilde{Z}_{i,j}}&\text{if $\reward=\reward^{(i)}_j$},\\
        F_{\action_i}(\reward)=\widetilde{e}^{-k} \qquad & \text{if } r \in R_{\action_i}\backslash\{\reward^{(i)}_0,\reward^{(i)}_j\},\\
        0 &\text{otherwise}.
    \end{cases}
\end{align*}
That is, the belief $\belief_{(j,\action_i)}$ only differs from the true underlying distribution $F_{a_i}$ in reward outcomes $0$ and $j$.

Next we observe that,
\begin{align*}
    \KL{F_{\action_i}}{\belief_{(j,\action_i)}} =& \InParentheses{1-n\cdot \widetilde{e}^{-k}} \log\InParentheses{\frac{1-n\cdot \widetilde{e}^{-k}}{\widetilde{e}^{-\widetilde{Z}_{i,j}}}} + \widetilde{e}^{-k}\cdot  \log\InParentheses{\frac{\widetilde{e}^{-k}}{1-(n-1)\cdot \widetilde{e}^{-k}-\widetilde{e}^{-\widetilde{Z}_{i,j}}}} \\
    =& \widetilde{Z}_{i,j} +  \widetilde{Z}_{i,j}\cdot (\log\InParentheses{\widetilde{e}}-1) -n\cdot \widetilde{e}^{-k}\cdot \widetilde{Z}_{i,j}\cdot \log(\widetilde{e}) - k\cdot \widetilde{e}^{-k} \log(\widetilde{e})
\\
    &\ + \InParentheses{1-n\cdot \widetilde{e}^{-k}} \log\InParentheses{1-n\cdot \widetilde{e}^{-k}}- \widetilde{e}^{-k}\cdot  \log\InParentheses{1-(n-1)\cdot \widetilde{e}^{-k}-\widetilde{e}^{-\widetilde{Z}_{i,j}}}.
\end{align*}
By \Cref{eq:bound on tilde first}-\Cref{eq:bound on tilde last}, and using the inequality $(1-x)\cdot \log(1-x)\geq -x$ we get:
\begin{align*}
    \widetilde{e}^{-k}\cdot \log\InParentheses{1-(n-1)\cdot\widetilde{e}^{-k}-\widetilde{e}^{-\widetilde{Z}_{i,j}}}\geq \widetilde{e}^{-k}\cdot \log\InParentheses{\frac{1}{4}}> -2 \cdot \widetilde{e}^{-k} ,\\
    \InParentheses{1-n\cdot \widetilde{e}^{-k}} \log\InParentheses{1-n\cdot \widetilde{e}^{-k}} \geq   - n \cdot \widetilde{e}^{-k}, \\
    k\cdot \widetilde{e}^{-k} \log(\widetilde{e})\leq \widetilde{e}^{-\frac{k}{2}}\log(\widetilde{e}) \qquad\qquad (k\geq 4)
\end{align*}
Using the inequalities above, we have the following simplification:
\begin{align*}
    \InAbsolute{\KL{F_{\action_i}}{\belief_{(j,\action_i)}} - \widetilde{Z}_{i,j}} \leq & \widetilde{Z}_{i,j}\cdot \frac{2^{-k}}{e} + n \cdot \widetilde{e}^{-k} \cdot \widetilde{Z}_{i,j} \cdot \log(\widetilde{e}) + \widetilde{e}^{-\frac{k}{2}}\log(\widetilde{e}) + n \cdot \widetilde{e}^{-k} + 2\cdot \widetilde{e}^{-k} \\
    \leq & \widetilde{Z}_{i,j}\cdot \frac{2^{-k}}{e} +  10 \cdot n\cdot \widetilde{Z}_{i,j} \cdot \widetilde{e}^{-k/2}.
\end{align*}
    Now we define the contract $\contract$ for rewards in $\rewards_{\action_i}$:
    \begin{align*}
        \contract(\reward^{(i)}_j)=
        \begin{cases}
            0 \qquad&\text{if $j=0$,}\\
            \frac{\widetilde{Y}_{i,j}}{1-(n-1)\cdot \widetilde{e}^{-k} - \widetilde{e}^{-\widetilde{Z}_{i,j}}} &\text{otherwise.}
        \end{cases}
    \end{align*}

    Under belief $\belief_{j}$, using \Cref{eq:bound on tilde first}-\Cref{eq:bound on tilde last}, the expected reward of the agent for action $\action_i$ is,
    \begin{align*}
        & \Util(\action_i,\belief_{j},\contract) = \widetilde{Y}_{i,j} + \widetilde{e}^{-k}\sum_{j'\in[n]\neq j} \frac{\widetilde{Y}_{i,j'}}{1-(n-1)\cdot \widetilde{e}^{-k} - \widetilde{e}^{-\widetilde{Z}_{i,j'}}} \\        \Rightarrow& \InAbsolute{\Util(\action_i,\belief_{i},\contract) -\widetilde{Y}_{i,j}} \leq \widetilde{e}^{-k}\sum_{j'\in[n]} \frac{\widetilde{Y}_{i,j'}}{1-(n-1)\cdot \widetilde{e}^{-k} - \widetilde{e}^{-\widetilde{Z}_{i,j'}}} 
        \leq 4\cdot \widetilde{e}^{-k}\sum_{j'\in[n]} \widetilde{Y}_{i,j'},
    \end{align*}

    Thus by setting $k=\polylog\InParentheses{\max_{j\in[n]}{\widetilde{Y}_{i,j'}}, \max_{j'\in[n]} {\widetilde{Z}_{i,j'}}, n,\frac{1}{\epsilon'}}$, for each belief $\belief_{j}$ we get,
    \begin{align*}
        \InAbsolute{\Util(\action_i,\belief_{j},\contract) - \widetilde{Y}_{i,j} } \leq \epsilon',\\
        \InAbsolute{\KL{F_{\action_i}}{\belief_{(j,\action_i)}} - \widetilde{Z}_{i,j}}\leq \epsilon'.
\end{align*}
Moreover, the bit complexity for the set of beliefs and the true distribution is: 
$$ \poly( n, \max_{j'\in[n]} \widetilde{Z}_{i,j'} ,k)= \poly\InParentheses{ n , \max_{j'\in[n]} \widetilde{Z}_{i,j'} , \max_{j'\in[n]}\log(\widetilde{Y}_{i,j'}), \log\InParentheses{\frac{1}{\epsilon'}}},$$
and the bit complexity of the contract $\contract$ is 
$$\poly(n,\max_{j'\in[n]}\log(\widetilde{Y}_{i,j'}), \max_{j'\in[n]}\widetilde{Z}_{i,j'}, k )=\poly\InParentheses{n,\max_{j'\in[n]}\log(\widetilde{Y}_{i,j'}), \max_{j'\in[n]}\widetilde{Z}_{i,j'}, \log\InParentheses{\frac{1}{\epsilon'}} }.$$
\end{proof}


For an additional parameter $\kappa=7$, define matrices \(\widetilde{Y},\widetilde{Z} \in \reals^{n}\times \reals^{n}\) such that for all \(i, j\in[n]\),
\[
\widetilde{Y}_{i,j} = \left\lfloor \frac{Y_{i,j}}{\frac{\epsilon^*}{\kappa}} \right\rfloor+1 \text{, and }\widetilde{Z}_{i,j} = \left\lfloor \frac{Z_{i,j}}{\frac{\epsilon^*}{\kappa}} \right\rfloor +1,
\]
where $\epsilon^*$ is the constant promised in \Cref{thm:inapproximability}.

\begin{observation}\label{obs:new nash}
An \(\epsilon'\)-Nash equilibrium $(x,y)$ in the two-player general-sum game $(\widetilde{Y},\widetilde{Z})$ is an \(\InParentheses{\epsilon' + 1}\cdot \frac{\epsilon^*}{\kappa}\)-Nash equilibrium in the two-player general-sum game $(Y,Z)$.
\end{observation}
\begin{proof}
    The reduction follows by calculations. Since $(x,y)\in \Delta([n])^2$ is an $\epsilon'$-Nash equilibrium of the general-sum game $(\widetilde{Y},\widetilde{Z})$, then:
    \begin{align*}
        \sum_{i,j\in[n]}x_i \InParentheses{\frac{Y_{i,j}}{\frac{\epsilon^*}{\kappa}} + 1} y_j \geq \sum_{i,j\in[n]}x_i \widetilde{Y}_{i,j}y_j \geq  \sum_{i,j\in[n]}x_i \widetilde{Y}_{i,j} - \epsilon' \geq \sum_{i\in[n]}x_i \frac{Y_{i,j}}{\frac{\epsilon^*}{\kappa}} - \epsilon' \qquad\qquad \forall j\in[n]\\
        \Rightarrow \sum_{i,j\in[n]}x_i \InParentheses{\frac{Y_{i,j}}{\frac{\epsilon^*}{\kappa}} + 1} y_j \geq \sum_{i\in[n]}x_i \frac{Y_{i,j}}{\frac{\epsilon^*}{\kappa}}  -\epsilon'\qquad\qquad \forall j\in[n]\\
\Rightarrow  \sum_{i,j\in[n]}x_i Y_{i,j} y_j \geq \sum_{i\in[n]}x_i Y_{i,j} - \frac{\epsilon}{\kappa} - \frac{\epsilon'\epsilon^*}{\kappa} \qquad\qquad \forall j\in[n].
    \end{align*}
    We can similarly show that:
    $$
    \sum_{i,j\in[n]}x_i Z_{i,j} y_j \geq \sum_{j\in[n]} Z_{i,j}y_j - \frac{\epsilon}{\kappa} - \frac{\epsilon'\epsilon^*}{\kappa}, \qquad\qquad \forall j\in[n].
    $$
    which concludes the proof.
\end{proof}

Now we consider the contract $\contract$ for the mispeciffied agent with beliefs $\beliefs$, actions $\actions$, and true distribution of rewards over actions $F$, as promised by \Cref{lem:red game to contract} for matrices $\widetilde{Y},\widetilde{Z}$ and precision parameter $1$. By construction, for a given belief $\belief_j\in \beliefs$, action $\action_i\in \actions$ we have:
\begin{align}
\InAbsolute{V(\action_i,\belief_{j},\contract) - \widetilde{Y}_{i,j}} \leq 1,\label{eq:item 1}\\
\InAbsolute{\KL{F_{\action_i}}{\belief_{(j,\action_i)}} - \widetilde{Z}_{i,j}}\leq 1 \label{eq:item 2}.
\end{align}

Notice that the entries of matrices $\widetilde{Y},\widetilde{Z}$ are integers and lie in the range $[1,\frac{\kappa}{\epsilon^*}+1]$, which further implies that the bit complexity of the contract $\contract$ and the set of beliefs is $\poly(n)$ ($\epsilon^*\leq 1$ is the constant promised in \Cref{thm:inapproximability}, and $\kappa=7$, which implies that each entry  of $\widetilde{Y},\widetilde{Z}$ has constant bit-complexity).

Now assume towards contradictions that under \Cref{hyp:eth} we can compute an $1$-Berk-Nash equilibrium $\adist^*\in \Delta(\actions)$ of the corresponding misspecified contract design instance in time $o\InParentheses{n^{\log^{1-o(1)}(n)}}$.
\footnote{Observe that in the Berk-Nash equilibrium instance, the values of $\contract$ are not normalized in $[0,1]$, but still bounded by an absolute constant.} Then $\adist^*$ has bit-complexity $o\InParentheses{n^{\log^{1-o(1)}(n)}}$ and there exists a posterior over beliefs $\posterior^*$ such that:

\begin{align*}
\sum_{i,j\in[n]}\adist^*(\action_i)\Util(\action_i,\belief_j,\contract)\posterior^*(\belief_j) \geq  \sum_{j\in[n]}{\Util(\action_i,\belief_j,\contract)}\posterior^*(\belief_j) - 1\qquad\qquad \forall i\in[n],\\
\sum_{i\in[n],j\in[n]} \adist^*(\action_i) \KL{F_{a_i}}{\belief_{(j,a_i)}}\posterior^*(\belief_j) \geq \sum_{i\in[n]}\adist^*(\action_i){\Util(\action_i,\belief_j,\contract)}\qquad\qquad \forall j\in[n]
\end{align*}

\sloppy
We overcharge notation and denote also by $\adist^*=[\adist^*(\action_1),\ldots,\adist^*(\action_n)]$, and $\posterior^*=[\posterior^*(\belief_1),\ldots,\posterior^*(\belief_n)]$. Combining the condition above with \eqref{eq:item 1}, \eqref{eq:item 2}, and the fact that $\|\adist^*\|_1=\|\posterior^*\|_1=1$ we infer that:
\begin{align*}
\adist^{*T} \widetilde{Y} \posterior^* \geq  e_i^T \widetilde{Y}\posterior^* - 3\qquad\qquad \forall i\in[n],\\
\adist^{*T} \widetilde{Z} \posterior^* \geq  \adist^{*T} \widetilde{Y}e_j - 2, \qquad\qquad\forall j\in[n],
\end{align*}
which implies the existence of $3$-Nash equilibrium on two-player general-sum game $(\widetilde{Y},\widetilde{Z})$. An obstacle is that potentially $\adist^*$ can have $o\InParentheses{n^{\log^{1-o(1)}(n)}}$ bit-complexity. To deal with that, we round $\adist^*$ to $\widetilde{\adist}^*$ with $\poly(n)$ bit-complexity as follows:
\begin{align*}
    \widetilde{\adist}^*(\action_i) = \begin{cases}
        \left\lfloor \frac{\adist^*(\action_i)}{\frac{\epsilon}{n^2\cdot \kappa}}\right\rfloor\cdot \frac{\epsilon}{n^2\cdot \kappa} \qquad\qquad &\forall i \in[n-1],\\
        1 - \sum_{j\in[n-1]} \widetilde{\adist}^*(\action_j) & \text{for }i=n.
    \end{cases}.
\end{align*}

An important observation is that for all $i\in[n]$, $\widetilde{\adist}^*(\action_i)\in [\adist^*(\action_i)-\frac{\epsilon}{n\cdot \kappa}, \adist^*(\action_i) + \frac{\epsilon}{n\cdot \kappa}]$. Thus,  \begin{align*}
\widetilde{\adist}^T \widetilde{Y} \posterior^*  \geq \adist^{*T} \widetilde{Y} \posterior^* - \frac{\epsilon}{\kappa}\max_{i,j\in[n]}\widetilde{Y}_{i,j} \geq  e_i^T \widetilde{Y}\posterior^* - 5\qquad\qquad \forall i\in[n],
\end{align*}
where we used that $\max_{i,j\in[n]}\widetilde{Y}_{i,j} \leq \frac{\kappa}{\epsilon} + 1 \leq 2\frac{\kappa}{\epsilon}$. Similarly we can prove that: 
\begin{align*}
    \widetilde{\adist}^{*T} \widetilde{Z} \posterior^* \geq \adist^{*T} \widetilde{Z} \posterior^* -2 \geq  \adist^{*T} \widetilde{Y}e_j - 4\geq \widetilde{\adist}^{*T} \widetilde{Y}e_j - 6, \qquad\qquad\forall j\in[n],
\end{align*}

Thus $(\widetilde{\adist}^*,\posterior^*)$ is a $6$ Nash equilibrium of the two-player general-sum game $(\widetilde{Y},\widetilde{Z})$ and $\widetilde{\adist}^*$ has $\poly(n)$ bit-complexity. Therefore, the following LP is guaranteed to have a solution and we can compute vector $\widetilde{\posterior}^*\in\Delta([n])$ in $\poly(n)$ time (each coordinate of $\widetilde{Y},\widetilde{Z}$ requires constant-size and the bit-complexity of $\widetilde{\adist}^*$ is $\poly(n)$):
    \begin{align*}
        \text{find } \widetilde{\posterior}^*\in \Delta([n]) \\
       \text{such that: }&\widetilde{\adist}^{*T} \widetilde{Y} \widetilde{\posterior}^*\geq \widetilde{\adist}^{*T} \widetilde{Y} e_{j} -6\qquad & \forall j\in[n],\\ 
& \widetilde{\adist}^{*T} \widetilde{Y} \widetilde{\posterior}^*\geq e_{i} \widetilde{Z} \widetilde{\posterior}^* -6 \qquad & \forall i\in[n].
    \end{align*}

Thus $(\widetilde{\adist}^*,\widetilde{\posterior}^*)$ is a $6$-Nash equilibrium of the two-player general-sum game $(\widetilde{Y},\widetilde{Z})$.
Moreover by \Cref{obs:new nash}, for $\kappa=7$, $(\widetilde{\adist}^*,\widetilde{\posterior}^*)$ is an $\epsilon$-nash equilibrium of the two-player general-sum game $(Y,Z)$. A contradictions under \Cref{hyp:eth}, since the total computational time is $o\InParentheses{n^{\log^{1-o(1)}(n)}}$ (\Cref{thm:inapproximability}).
\end{prevproof}

%% file: content/divergence_under_mispecification.tex
\section{Ergodic Divergence Under Misspecification with Three Actions}
\label{apx:ergodic_divergence}
In this section we show that the empirical frequency of the belief and the action of a misspecified agent does not converges ergodically in the limit even when the agent is correctly specified for each of her actions under some of her beliefs. 
This claim is made in \citet{esponda2021asymptotic} without proofs. 
In particular, 
\citet{esponda2021asymptotic} consider a learning agent and in their Example~2, they show dynamics that display ergodic divergence. In their example, the agent is agnostic to the effect her action has on the system, and tries to fit her observations into a normal distribution. A notable difference with our example is that in our construction the agent chooses the action that maximizes her posterior, while in their construction the agent would be indifferent between actions and selects actions based on a predefined correspondence between beliefs and actions.\footnote{In their construction, the agent does not model the effect her action has on the outcome, and as such, the utility between different actions would be equal.} 
Thus their result connot imply the ergodic divergence in environments where the action taken by the Bayesian learning agent is determined by the observable payoff-relevant outcomes. 

We provide the full details regarding the constructions and proofs for ergodic divergence in this appendix. 
Throughout this section, we consider the following contract for an agent with three possible actions and misspecified beliefs.

    \begin{instance}\label{ins:divergence}
    We consider the following instance of contract design between a principal and an agent:
    \begin{itemize}
        \item An agent with three available actions $\actions = \InBraces{\action_0, \action_1, \action_2}$.
        \item There are four possible rewards $\rewards=\InBraces{\reward_0, \reward_1, \reward_2, \reward_f}$ for the principal, and the agent has the following contract:  $$\contract=\InBraces{\contract(\reward_0)=0, \contract(\reward_1)=1,\contract(\reward_2)=1,\contract(\reward_f)=0}.$$ 
        \item For each action the agent picks, the produced reward is deterministic, as more specifically: $$F_{\action_i}=\delta_{\reward_i} \qquad\qquad\qquad i\in \InBraces{1,2,3}.$$ 
        Notice that reward $\reward_f$ is fictitious and is not a valid outcome from an action.
    \end{itemize}

    The agent is misspecified and has three possible beliefs about the outcome of her actions $\beliefs = \InBraces{\belief^0, \belief^1, \belief^2}$, defined as follows for action $\action_i$ for $i\in \InBraces{0,1,2}$:

            \begin{align*}
            \belief_i(\action_i,\reward)=&\ind{\reward=\reward_i},\\
            \belief_{i+1 \mod{3}}(\action_i,\reward)=&\begin{cases}
                0.125 \qquad & \text{if $\reward = \reward_i$} \\ 0.875 \qquad & \text{if $\reward = \reward_{i+1 \mod{3}}$}
            \end{cases}, \\
            \belief_{i+2 \mod{3}}(\action_i,\reward)=&\begin{cases}
                0.25 \qquad & \text{if $\reward = \reward_{i }$} \\ 0.75 \qquad & \text{if $\reward = \reward_f$}
            \end{cases}, 
            \end{align*}

    Observe that under belief $\belief_i$ the agent is correctly specified for the outcome of action $\action_i$. We are going to show that despite this seemingly correct specification of the agent, the agent's belief is going to cycle the space of beliefs perpetually without converging ergodically. 

    We either index actions by the time they were played using subscript $t$, e.g., $a_t$ is the action the agent used at round $t$, or by number, e.g., $a_w$ where $w\in \{1,2,3\}$. It will be clear from content in  which case we refer to.
    \end{instance}

    \begin{observation}\label{obs:picks action}
        Assume that the agent has posterior $\posterior_t \in\Delta(\beliefs)$ over beliefs at time $t$, then the action that maximizes the agent's utility based can be expressed as:
        $$\argmax_{i\in\InBraces{0,1,2}} \expect[\belief\sim \posterior_t ]{V(\action_i,\belief,\contract)}= \argmin_{i\in\InBraces{0,1,2}}\posterior_t(\belief_{i+1\mod{3}}).$$
    \end{observation}

\Cref{lem:multiplicative cycles} is the building block for our proof. At a high level, we demonstrate the following key aspects within \Cref{ins:divergence}: Firstly, an agent employing best-response dynamics will indefinitely alternate her chosen actions (as outlined in the first item). Secondly, this alteration follows a distinct pattern: the agent shifts from action \(\action_w\) to \(\action_{w-1 \mod 3}\), then moves from \(\action_{w-1 \mod 3}\) to \(\action_{w-2 \mod 3}\), and subsequently returns to \(\action_w\) (as detailed in the second item). Lastly, we establish that the agent settles on a consistent action within time windows that display an exponentially increasing frequency (described in the third item).


\begin{lemma}\label{lem:exp cycle}
    Consider a misspecified learning agent in \Cref{ins:divergence} who uses best-response dynamics and break ties lexicographically.
    Assume at time $t$, the agent picks action $\action_{j^*}$ and let $[T,T'-1]$ be the maximal interval that the agent picks consecutively action $\action_{j^*}$. Let $k^*=j^*+1\mod{3}$ and $i^*=j^*-1\mod{3}$, then the following are true:
    \begin{itemize}
        \item The agent picks action $\action_{i^*}$ at round $T'$.
        \item $i^*\in \argmin_{i\in\InBraces{0,1,2}}\posterior_T( \belief_i)$.
        \item $T'-T \in \InBrackets{\log_2\InParentheses{\frac{\posterior_T(\belief_{k^*})}{\posterior_T(\belief_{i^*})}},\log_2\InParentheses{\frac{\posterior_T(\belief_{k^*})}{\posterior_T(\belief_{i^*})}}+1}$.
        \item $\log_2\InParentheses{\frac{\posterior_{T'}(\belief_{j^*})}{\posterior_{T'}(\belief_{k^*})}} \geq 2\log_2\InParentheses{\frac{\posterior_T(\belief_{k^*})}{\posterior_T(\belief_{i^*})}}$.
        \item In addition, when $\log_2\InParentheses{\frac{\posterior_T(\belief_{k^*})}{\posterior_T(\belief_{i^*})}}\geq 8$, then $\log_2\InParentheses{\frac{\posterior_{T'}(\belief_{j^*})}{\posterior_{T'}(\belief_{k^*})}}  \geq 1.5\cdot (T'-T)$.
    \end{itemize}
\end{lemma}

\begin{prevproof}{lem:exp cycle}
     By \Cref{obs:picks action} we know that $j^*= i^*+1 \mod{3}$. For any time $t\in [T,T'-1]$, that the agent picks action $\action_{j^*}$, and since the outcome reward is deterministically $r_{j^*}$, she updates her belief according to the bayes rule as follows:

    \begin{align*}
    \posterior_{t+1}(\belief_{j^*}) \propto& \posterior_{T}(\belief_{j^*}), \\ 
    \posterior_{t+1}(\belief_{k^*}) \propto& \posterior_{T}(\belief_{k^*})\cdot 0.125^{t+1-T}, \\ 
    \posterior_{t+1}(\belief_{i^*}) \propto& \posterior_{T}(\belief_{i^*})\cdot 0.25^{t+1-T}. 
    \end{align*}

    By \Cref{obs:picks action}, since the agent picked action $j^*$ at time $T$, then $\posterior_T(\belief_{i^*}) \leq\posterior_T(\belief_{j^*})$, $\posterior_T(\belief_{k^*})$. By the update rule, it is clear that for any $\frac{\posterior_{t+1}(\belief_{j^*})}{ \posterior_{t+1}(\belief_{i^*})}> 1$, which in combination with \Cref{obs:picks action}, further implies that the critical event for the player to change her action is when $k^* \in \argmin_{i\in\InBraces{0,1,2}}\posterior_{T'}(\belief_{i})$, for which a sufficient condition is:

    \begin{align*}
    \posterior_{T'}(\belief_{k^*}) \leq \posterior_{T'}(\belief_{i^*})
    &\Leftrightarrow    \posterior_{T}(\belief_{k^*})\cdot 0.125^{T'-T}
    \leq \posterior_{T}(\belief_{i^*})\cdot 0.25^{T'-T} \\
    &\Leftrightarrow T'-T \geq \uceil{\log_2\InParentheses{\frac{\posterior_T(\belief_{k^*})}{\posterior_T(\belief_{i^*})}}}.
    \end{align*}
    \sloppy
    Observe that by \Cref{obs:picks action}, the only scenario when at time $T''=T + \uceil{\log_2\InParentheses{\frac{\posterior_T(\belief_{k^*})}{\posterior_T(\belief_{i^*})}}}$, the player still picks action $\action_{j^*}$ rather than action $\action_{i^*}$, is when  $\posterior_{T''}(\belief_{k^*}), \posterior_{T''}(\belief_{i^*}) \in \argmin_{i\in \InBraces{0,1,2}}\posterior_{T''}(\belief_i)$, and the tie is lost in favor of $i^*$. However, in this case it must be the case that $\log_2\InParentheses{\frac{\posterior_T(\belief_{k^*})}{\posterior_T(\belief_{i^*})}}$ is an integer and at round $T''+1$,  $\posterior_{T''}(\belief_{k^*})$ is the unique element in $ \argmin_{i\in \InBraces{0,1,2}}\posterior_{T''+1}(\belief_i)$, and the player picks action $i^*$ according to \Cref{obs:picks action}. Thus $T'-T \in \InBrackets{\log_2\InParentheses{\frac{\posterior_T(\belief_{k^*})}{\posterior_T(\belief_{i^*})}}, \log_2\InParentheses{\frac{\posterior_T(\belief_{k^*})}{\posterior_T(\belief_{i^*})}}+1}$. Finally,

    \begin{align*}
        \frac{\posterior_{T'}(\belief_{j^*})}{\posterior_{T'}(\belief_{k^*})} = \frac{\posterior_T(\belief_{j^*})}{\posterior_T(\belief_{k^*})}\cdot 8^{T'-T}\geq \frac{\posterior_T(\belief_{j^*})}{\posterior_T(\belief_{k^*})}\cdot \frac{\posterior_T(\belief_{k^*})^3}{\posterior_T(\belief_{i^*})^3}=\frac{\posterior_T(\belief_{k^*})^2}{\posterior_T(\belief_{i^*})^2}\cdot \frac{\posterior_T(\belief_{j^*})}{\posterior_T(\belief_{i^*})}\geq \InParentheses{\frac{\posterior_T(\belief_{k^*})}{\posterior_T(\belief_{i^*})}}^2,
    \end{align*}
    where the inequalities follows by the lower bound on $T'-T$ and the fact that $i^*\in \argmin_{i\in\InBraces{0,1,2}}\posterior_t(\belief_i)$. The third item of the statement is proven y taking the logarithm of the expression above. We now prove the final item of the statement assuming that $\log_2\InParentheses{\frac{\posterior_T(\belief_{k^*})}{\posterior_T(\belief_{i^*})}}\geq 8$ it suffices to prove:

    \begin{align*}
&2\log_2\InParentheses{\frac{\posterior_T(\belief_{k^*})}{\posterior_T(\belief_{i^*})}}\geq  1.5 (T'-T)\\
\Leftrightarrow& 1.5\cdot\log_2\InParentheses{\frac{\posterior_T(\belief_{k^*})}{\posterior_T(\belief_{i^*})}} + \frac{1}{2}\log_2\InParentheses{\frac{\posterior_T(\belief_{k^*})}{\posterior_T(\belief_{i^*})}} \geq 1.5 (T'-T)\\
\Rightarrow &1.5\cdot\log_2\InParentheses{\frac{\posterior_T(\belief_{k^*})}{\posterior_T(\belief_{i^*})}} + \frac{1}{2}8 \geq 1.5\cdot\InParentheses{\log_2\InParentheses{\frac{\posterior_T(\belief_{k^*})}{\posterior_T(\belief_{i^*})}} +1}\geq 1.5 (T'-T),
    \end{align*}
where we used the assumption that $\log_2\InParentheses{\frac{\posterior_T(\belief_{k^*})}{\posterior_T(\belief_{i^*})}}\geq 8$ and concluded the proof by applying item 2
\end{prevproof}

\begin{lemma}\label{lem:multiplicative cycles}
    Consider a misspecified learning agent in \Cref{ins:divergence} who uses best-response dynamics and break ties lexicographically such that at time $1$, the agent picks action $\action_{j^*}$ and let  $k^*=j^*+1\mod{3}$ and $i^*=j^*-1\mod{3}$. Further assume that $\log_2\InParentheses{\frac{\posterior_1(\belief_{k^*})}{\posterior_1(\belief_{i^*})}}\geq 8$ and consider the critical times $\mathcal{T}_C=\InBraces{t_0,t_1,t_2,\ldots,t_n,\ldots}$ where the agent changed her action, e.g.:
    \begin{align*}
        t_0 =& 1,\\
        t_n =& \argmin_{t> t_{n-1}}\mathds{1}[a_t\neq a_{t-1}].
    \end{align*}
    
    Then the following are true:
    \begin{itemize}
        \item The agent changes her action infinitely many times e.g. $|\mathcal{T}_C|=+\infty$ and chooses each action infinitely many times e.g. for any action $\action\in \actions$, $|t_n\in \mathcal{T}_C: \action_{t_n}=\action |=+\infty$.
        \item If at time $t_n$ the agent picks action $\action_{w}$, then at round $t_{n+1}$ the agent picks action $\action_{w-1\mod{3}}$.
        \item Moreover for any $n\geq 0$: $(t_{n+2}-t_{n+1}) \geq 1.5\cdot (t_{n+1}-t_n)$.
    \end{itemize}
\end{lemma}

\begin{prevproof}{lem:multiplicative cycles}
    The first item of the statement follows by inductive use the third item in \Cref{lem:exp cycle}. Formally, assume that there exists a finite time $t^*$ such that the agent picks action $\action_w$ at $t^*$ and then for any other time $t>t^*$, $\action_t=\action_w$. Since $t$ is finite, then $\log_2\InParentheses{\frac{\posterior_{t^*}(\belief_{w+1\mod{3}})}{\posterior_{t^*}(\belief_{w-1\mod{3}})}}$ is also finite, which by item three of \Cref{lem:exp cycle} implies that the agent changes her action in at most $\log_2\InParentheses{\frac{\posterior_{t^*}(\belief_{w+1\mod{3}})}{\posterior_{t^*}(\belief_{w-1\mod{3}})}}+1$ rounds.

    The second item of the statement follows by the first item of \Cref{lem:exp cycle}.

    We show the third item of the statement using induction. More specifically we prove with induction on $n$ that if the agent at round $t_n$ picks action $a_w$ that satisfies:
    \begin{align*}
        \log_2\InParentheses{\frac{\posterior_{t_n}(\belief_{w+1\mod{3}})}{\posterior_{t_n}(\belief_{w-1\mod{3}})}}\geq 8,
    \end{align*}
    then at round $t_{n+1}$ the agent picks action $\action_{w'}$, where $w'=w-1\mod{3}$ that satisfies:
    \begin{align*}
        \log_2\InParentheses{\frac{\posterior_{t_{n+1}}(\belief_{w'+1\mod{3}})}{\posterior_{t_{n+1}}(\belief_{w'-1\mod{3}})}}\geq 8
    \end{align*}
    \paragraph{Base case $n=0$:} 
    The base case holds trivially by assumption.

    \paragraph{Induction Step:} Assume that at time $t_n$, agent picks action $a_w$, and let $w'=w-1\mod{3}$. Then by \Cref{lem:exp cycle} agent picks action $a_{w-1\mod{3}}$ at round $t_{n+1}$. By item four in \Cref{lem:exp cycle} we further have:
    \begin{align*}
        \log_2\InParentheses{\frac{\posterior_{t_{n+1}}(\belief_{w'+1\mod{3}})}{\posterior_{t_{n+1}}(\belief_{w'-1\mod{3}})}}=\log_2\InParentheses{\frac{\posterior_{t_{n+1}}(\belief_{w})}{\posterior_{t_{n+1}}(\belief_{w+1\mod{3}})}} \geq 2\log_2\InParentheses{\frac{\posterior_{t_n}(\belief_{w+1\mod{3}})}{\posterior_{t_n}(\belief_{w-1\mod{3}})}} \geq 16.
    \end{align*}

    Finally we show that if at round $n$, agent picks action $a_w$ such that $\log_2\InParentheses{\frac{\posterior_{t_n}(\belief_{w+1\mod{3}})}{\posterior_{t_n}(\belief_{w-1\mod{3}})}}\geq 8$, then we have that 
    $$
    t_{n+2} - t_{n+1}\geq 1.5\cdot (t_{n+1}-t_n).$$

    By last item in \Cref{lem:exp cycle} we have that:
    \begin{align*}
        \log_2\InParentheses{\frac{\posterior_{t_{n+1}}(\belief_{w})}{\posterior_{t_{n+1}}(\belief_{w+1\mod{3}})}}  \geq& 1.5\cdot  (t_{n+1}-t_n).
    \end{align*}
    By applying item one in \Cref{lem:exp cycle}, the agent picks action $\action_{w'}$, where $w'=w-1\mod{3}$. By applying item three in \Cref{lem:exp cycle} at times $t_{n+1}$ we further have:
    \begin{align*}
        t_{n+2} - t_{n+1} \geq& \log_2\InParentheses{\frac{\posterior_{t_{n+1}}(\belief_{w})}{\posterior_{t_{n+1}}(\belief_{w+1\mod{3}})}}    
    \end{align*}
    Chaining the two inequalities we prove that $t_{n+2} - t_{n+1}\geq 1.5\cdot (t_{n+1} - t_n)$.
\end{prevproof}

In \Cref{thm:ergodic divergence}, we formally show how to use \Cref{lem:exp cycle} to prove that the frequency of the agent's actions does not converge in an ergodically.

\begin{theorem}\label{thm:ergodic divergence}
     Consider a misspecified learning agent in \Cref{ins:divergence} who uses best-response dynamics and break ties lexicographically.
     Then for any action $\action \in \actions$ the following holds:
     $$
     \lim_{T\rightarrow +\infty}\inf \frac{\sum_{t\in[T]}\ind{a_t=a}}{t}< \lim_{T\rightarrow +\infty}\sup \frac{\sum_{t\in[T]}\ind{a_t=a}}{t}.$$

\end{theorem}

\begin{prevproof}{thm:ergodic divergence}
    We first show that at some time $t\in[30]$, the agent picks picks action $\action_{j^*}$ that satisfies $\log_2\InParentheses{\frac{\posterior_t(\belief_{j^*+1\mod{3}})}{\posterior_t(\belief_{j^*-1\mod{3}})}}\geq 8$.
    
    Let $\action_{j^*}$ be the action that the agent picks at round $0$. If $\log_2\InParentheses{\frac{\posterior_t(\belief_{j^*+1\mod{3}})}{\posterior_t(\belief_{j^*-1\mod{3}})}}\geq 8$ then we are done.
    Assume now that $\log_2\InParentheses{\frac{\posterior_t(\belief_{j^*+1\mod{3}})}{\posterior_t(\belief_{j^*-1\mod{3}})}} < 8$, and let $t_1<t_2<t_3$ be the times that the agent changed the action she used consecutively, that is for time $t\in[0,t_1-1]$ the agent used action $a_{j^*}$, at rounds $[t_1,t_2-1]$ the agent used consecutively a different action etc. By iterative use of \Cref{lem:exp cycle}, we know that the agent picked action $a_{j^*-1\mod{3}}$ at time $t_1$, action $a_{j^*-2\mod{3}}$ at time $t_2$ and again action $a_{j^*}$ at time $t_3$. By item three of \Cref{lem:exp cycle}, $t_1\leq \log_2\InParentheses{\frac{\posterior_0(\belief_{j^*+1\mod{3}})}{\posterior_0(\belief_{j^*-1\mod{3}})}} +1 \leq 9$, and if $\log_2\InParentheses{\frac{\posterior_{t_1}(\belief_{j^*})}{\posterior_{t_1}(\belief_{j^*+1\mod{3}})}}\geq 8$, then we are done. If $\log_2\InParentheses{\frac{\posterior_{t_1}(\belief_{j^*})}{\posterior_{t_1}(\belief_{j^*+1\mod{3}})}}< 8$, then by repeating the same argument $t_2 - t_1 \leq 9$. If $\log_2\InParentheses{\frac{\posterior_{t_2}(\belief_{j^*-1\mod{3}})}{\posterior_{t_2}(\belief_{j^*})}} > 8$, then we are done, if not then $t_3-t_2\leq 9$ and by iterative use of item four in \Cref{lem:exp cycle}, we have that:
\begin{align*}
    8 \leq  &8\log_2\InParentheses{\frac{\posterior_0(\belief_{j^*+1\mod{3}})}{\posterior_0(\belief_{j^*-1\mod{3}})}}  \qquad\qquad \text{(by \Cref{lem:exp cycle}, $j^*-1\mod{3}\in \argmin_{i\in\InBraces{0,1,2}}\posterior_{0}(\belief_i)$)}\\
    \leq& 4 \log_2\InParentheses{\frac{\posterior_{t_1}(\belief_{j^*})}{\posterior_{t_1}(\belief_{j^*+1\mod{3}})}} \qquad\qquad \text{(agent picks action $\action_{j^*}$ in rounds $[0,t_1)$)}\\
    \leq& 2 \log_2\InParentheses{\frac{\posterior_{t_2}(\belief_{j^*-1\mod{3}})}{\posterior_{t_2}(\belief_{j^*})}} \qquad\qquad \text{(agent picks action $\action_{j^*-1\mod{3}}$ in rounds $[t_1,t_2)$)} \\
    \leq& \log_2\InParentheses{\frac{\posterior_{t_3}(\belief_{j^*+1\mod{3}})}{\posterior_{t_3}(\belief_{j^*-1\mod{3}})}} \qquad\qquad \text{(agent picks action $\action_{j^*-2\mod{3}}$ in rounds $[t_2,t_3)$)} . 
\end{align*}
\sloppy
    Thus at some round $t^*\in[0,30]$, the agent picks picks action $\action_{j^*}$ that satisfies $\log_2\InParentheses{\frac{\posterior_{t^*}(\belief_{j^*+1\mod{3}})}{\posterior_{t^*}(\belief_{j^*-1\mod{3}})}}\geq 8$. 
    Let $\mathcal{T}_C=\InBraces{t_1=t^*\leq t_2,\ldots t_{n-1}\leq t_n\leq \cdots}$, be all times where the agent changed an action, formally $t_0 = t^*$ and for $n\geq 1$:
    
    $$t_n = \argmin_{t > t_{n-1}} \ind{\action_{t}\neq \action_{t-1}}.$$
    
   By \Cref{lem:multiplicative cycles} $|\mathcal{T}_C|=\infty$, thus the agent changes actions infinitely often. Now we show by induction that every for every $n\geq 1$,
   \begin{align}
       t_{n}-t_{n-1} \geq \frac{t_{n}-1}{3}.\label{eq:constant rounds}
   \end{align}
   
   The base case for $n=0$ holds trivially since by definition $t_0=1$. By item three in \Cref{lem:multiplicative cycles} and induction hypotheses:

   \begin{align*}
       t_{n} - t_{n-1} = \frac{t_n - t_{n-1}}{1.5} + \frac{t_n - t_{n-1}}{3} \geq t_{n-1}-t_{n-2} + \frac{t_n - t_{n-1}}{3} \geq \frac{t_{n-1}-1}{3} + \frac{t_n - t_{n-1}}{3} = \frac{t_n-1}{3}.
   \end{align*}

    Now we are ready to prove the main statement of the theorem. Assume towards contradiction that there exists action $\action_w$ such that:
    $$
     \lim_{T\rightarrow +\infty}\inf \frac{\sum_{t\in[T]}\ind{a_t=a_w}}{T}= \lim_{T\rightarrow +\infty}\sup \frac{\sum_{t\in[T]}\ind{a_t=a_w}}{T}=C.$$

    First we show that $C>\frac{1}{3}$. Let $t_n\in \mathcal{T}_C$ be a time such that $\action_{t_n}=\action_w$. We remind readers that during time period $[t_n,t_{n+1}-1]$ the agent picks action $\action_w$ consecutively. Then by \Cref{eq:constant rounds} we have that:
     $$
      \frac{\sum_{t\in[t_{n+1}]}\ind{a_t=a_w}}{t_{n+1}}\geq \frac{t_{n+1}-t_n}{t_{n+1}}\geq \frac{t_{n+1}-1}{3\cdot t_{n+1}}.$$

    Since by \Cref{lem:multiplicative cycles} the agent changes to action $\action_w$ infinitely many times, as $n\rightarrow+\infty$, then we have that:

    \begin{align}
    C=\lim_{T\rightarrow +\infty}\sup \frac{\sum_{t\in[T]}\ind{a_t=a_w}}{T}\geq \lim_{n\rightarrow +\infty: \action_{t_n}=\action_w}\sup \frac{\sum_{t\in[t_{n+1}]}\ind{a_t=a_w}}{t_{n+1}} \geq \frac{1}{3}.\label{eq:positive limit}
     \end{align}

    By \Cref{lem:multiplicative cycles}, $|\mathcal{T}_C|=+\infty$ and by \Cref{eq:constant rounds} we have that:
    $$
    \lim_{n\rightarrow+\infty}\frac{t_n}{t_{n+1}}\leq \frac{2}{3}
    $$

    Let $\mathcal{N}_w=\InBraces{n: \action_{t_n}\neq \action_w}$ be the indexes of times that the agent picked an action other than $\action_w$ in time $t_n$. By item one in \Cref{lem:multiplicative cycles}, $\|\mathcal{N}\|=+\infty$. We remind readers that for $n\in \mathcal{N}_w$, during time period $[t_n,t_{n+1}-1]$ the agent picks an action other than $\action_w$ constantly. Thus:
    \begin{align*}
    \frac{\sum_{t\in[t_{n+1}]}\ind{a_t=a_w}}{t_{n+1}}
    = &
      \frac{\sum_{t\in[t_{n}]}\ind{a_t=a_w}}{t_{n+1}} \\
    =& \frac{t_n}{t_{n+1}}\frac{\sum_{t\in[t_{n}]}\ind{a_t=a_w}}{t_{n}} 
    \end{align*}
    
    By assumption as $\lim n\in\rightarrow_\infty$:  $\frac{\sum_{t\in[t_{n}]}\ind{a_t=a_w}}{t_{n}}=C$.    
    Combining everything 
    \begin{align*}
C=\lim_{n\rightarrow+\infty:n \in \mathcal{N}_w}\frac{\sum_{t\in[t_{n+1}]}\ind{a_t=a_w}}{t_{n+1}}
    =& \sup\lim_{n\rightarrow+\infty:n \in \mathcal{N}_w}\frac{t_n}{t_{n+1}}\frac{\sum_{t\in[t_{n}]}\ind{a_t=a_w}}{t_{n}} \\
    =& \sup\lim_{n\rightarrow+\infty:n \in \mathcal{N}_w}\frac{t_n}{t_{n+1}}\sup\lim_{n\rightarrow+\infty:n \in \mathcal{N}_w}\cdot\frac{\sum_{t\in[t_{n}]}\ind{a_t=a_w}}{t_{n}} \\
    \leq & \frac{2}{3}\cdot C\\
    \Leftrightarrow & C = 0.
    \end{align*}
      A clear contradiction with respect to \Cref{eq:positive limit}.
\end{prevproof}

%% file: appendix/opt_contract.tex
\section{Missing Proofs for Optimal Contracts}
\label{apx:proof_opt_contract}
\subsection{Missing Proofs for Polynomial Time Algorithms}
\label{apx:poly_algo}
\begin{prevproof}{lem:sparce posterior}
Fix any distribution over actions $\adist^*=\InBraces{\adist^*(\action)}_{\action\in \actions}\in \Delta(\actions)$, and consider the set of beliefs that minimize the KL-divergence with respect to action profile $\adist^*$:
$$\belief'=\argmin_{\belief\in \beliefs} 
\expect[\action\sim\adist^*]{\KL{\dist_{\action}}{\belief_{\action}}}=\argmin_{\belief\in \beliefs} 
\inr{\adist^*,\kl(\belief)}.$$

Assume towards contradiction that $|\belief'|\geq 3$ which further implies that there exists $3$ distinct beliefs $\InBraces{\belief_0^*,\belief_1^*,\belief_{2}^*}\subseteq \beliefs$ such that
\begin{align*}
&\inr{\adist^*,\kl(\belief_1^*)} = \inr{\adist^*,\kl(\belief_0^*)},\\
&\inr{\adist^*,\kl(\belief_2^*)} = \inr{\adist^*,\kl(\belief_0^*)},\\
\Leftrightarrow  &  \inr{\adist^*,\kl(\belief_1^*)-\kl(\belief_0^*)} = 0,\\
&\inr{\adist^*,\kl(\belief_2^*)-\kl(\belief_0^*)} = 0.
\end{align*}
By \Cref{ass:poly contract}, set $\kls\subseteq\reals^{\actions}$ is in general position, which implies that vectors of beliefs $\InBraces{\belief_1^*-\belief^*_0,\belief_2^*-\belief^*_0}$ are linearly independent, thus forming a base for the space $\reals^{2}$. This implies that for any vector $v\in \reals^{2}$, $\inr{\adist^*,v}=0$ since $v$ can be written as linear combination of $\InBraces{\belief_1^*-\belief^*_0,\belief_2^*-\belief^*_0}$. Thus $\adist^*$ can only be the all-zero vector, a contradiction since $\adist^*\in \Delta(\actions)$.
\end{prevproof}

\begin{prevproof}{lem:charactere actions}
    The fact that the program is linear follows by observing that the only variables are $\adist^*(\action_1)$ and $\adist^*(\action_2)$, and the number of constraints are exactly $|\beliefs^*|\cdot|\beliefs|+3$. 
    
    The second part of the statement follows by observing that the set of distributions over actions $\adist^*\in \Delta(\actions)$ that satisfy:
    $$\sum_{\action\in \actions}\adist^*(\action)\cdot \InParentheses{{\KL{\dist_{\action}}{\belief^*_{\action}}} -  \KL{\dist_{\action}}{\belief_{\action}}} \leq 0 \qquad   \forall \belief^* \in \beliefs^*, \belief\in \beliefs,$$
    is a convex 1-dimensional line (since the feasibility program above is a linear program, and $\Delta(\actions)=\Delta(\reals^2)$ is a 1-dimensional line), with extreme points $\InBraces{\InParentheses{\action_1^{min},1-\action_1^{min}},\InParentheses{\action_1^{max},1-\action_1^{max}}}.$
\end{prevproof}

\begin{prevproof}{lem:union of convex sets}
We can rewrite the expression in the objective as:
\begin{align*}
 \contracts = \bigg\{\contract:\rewards\rightarrow\reals_+: \exists \posterior^*\in \Delta(\beliefs^*), & \sum_{\belief \in \belief^*} \posterior^*(\belief)\sum_{\reward\in\rewards} \belief_{\action^*}(\reward) \InParentheses{ \contract(\reward) -\cost(\action^*)}  
 \\ &
 \geq   \sum_{\belief \in \belief^*} \posterior^*(\belief)\sum_{\reward\in\rewards} \belief_\action(\reward) \InParentheses{ \contract(\reward) -\cost(\action)} \forall \action^* \in \actions^*, \action\in \actions \bigg\}.
 \end{align*}
The proof follows by taking cases for \(|\beliefs^*|=1\) (\Cref{lem:charactere contracts single belief}) and \(|\beliefs^*|=2\) (\Cref{lem:charactere contracts}). 

\begin{observation}\label{lem:charactere contracts single belief}
    For a belief $\belief^*\in \beliefs $ define the following quantity that only depend linearly on the contract $\contract(\cdot)$:

    \begin{align*}
    E(\contract):=& \sum_{\reward\in\rewards } \InParentheses{
    \InParentheses{ \contract(\reward) -\cost(\action_2)}\cdot  \belief^*_{\action_2}(\reward) - \InParentheses{ \contract(\reward) -\cost(\action_1)}\cdot  \belief^*_{\action_1}(\reward)  }
    \end{align*}

    A contract $\contract(\cdot)\in \reals_+^{|\rewards|}$ satisfies:
\begin{align*}
    \Util(\action_1,\belief^*,\contract) =&  \Util(\action_2,\belief^*,\contract)  
    \text{ iff } E(\contract) = 0,\\
     \Util(\action_1,\belief^*,\contract) \geq&  \Util(\action_2,\belief^*,\contract)  \text{
    iff }   E(\contract)\leq 0 , \\ 
     \Util(\action_1,\belief^*,\contract) \leq &  \Util(\action_2,\belief^*,\contract)  
   \text{ iff }  E(\contract) \geq 0. \end{align*}
\end{observation}
\begin{proof}
    The proof follows by expanding terms $\Util(\action_1,\belief^*,\contract)$ and $\Util(\action_2,\belief^*,\contract)$.
\end{proof}

\begin{lemma}\label{lem:charactere contracts}
Consider a set that contains at most two beliefs \(\beliefs^*=\InParentheses{\belief^{(*,1)},\belief^{(*,2)}} \subseteq \beliefs\). Define the following two quantities that only depend linearly on the contract $\contract(\cdot)$:

\begin{align*}
D(\contract):=& \sum_{\reward\in\rewards}\InParentheses{\InParentheses{ \contract(\reward) -\cost(\action_1)}\cdot  \InParentheses{\belief^{(*,1)}_{\action_1}(\reward) -  \belief^{(*,2)}_{\action_1}(\reward)} - \InParentheses{ \contract(\reward) -\cost(\action_2)}\cdot  \InParentheses{\belief^{(*,1)}_{\action_2}(\reward) -  \belief^{(*,2)}_{\action_2}(\reward)}} ,\\
E(\contract):=& \sum_{\reward\in\rewards } \InParentheses{
\InParentheses{ \contract(\reward) -\cost(\action_2)}\cdot  \belief^{(*,2)}_{\action_2}(\reward) - \InParentheses{ \contract(\reward) -\cost(\action_1)}\cdot  \belief^{(*,2)}_{\action_1}(\reward)  }.
\end{align*}    
For any contract $\contract(\cdot)\in \reals_+^{|\rewards|}$, there exists posterior over beliefs $\posterior^*\in \Delta(\beliefs^*)$ such that:
\begin{align*}
\expect[\belief\sim \posterior^*] { \Util(\action_1,\belief,\contract)} =&  \expect[\belief\sim \posterior^*] {\Util(\action_2,\belief,\contract) } 
\text{ iff } D(\contract)\geq E(\contract) \geq 0 \text{ or } 0 \geq E(\contract) \geq D(\contract),\\
\expect[\belief\sim \posterior^*] { \Util(\action_1,\belief,\contract)} \geq&  \expect[\belief\sim \posterior^*] {\Util(\action_2,\belief,\contract) }  \text{
iff }  D(\contract)\geq E(\contract)\geq 0 \text{ or }  E(\contract) \leq 0 , \\ \expect[\belief\sim \posterior^*] { \Util(\action_1,\belief,\contract)} \leq &  \expect[\belief\sim \posterior^*] {\Util(\action_2,\belief,\contract) } 
\text{ iff } D(\contract)\leq E(\contract) \leq 0 \text{ or }  E(\contract) \geq 0. \end{align*}
\end{lemma}

\begin{proof}
    \sloppy
    We only prove the first case of the statement as the rest follow in a similar way. Assume that there exists a posterior $\posterior^*\in \Delta(\belief^*)$ such that
    $
    \expect[\belief\sim \posterior^*] { \Util(\action_1,\belief,\contract)} =  \expect[\belief\sim \posterior^*] {\Util(\action_2,\belief,\contract) } $. By expanding the original expression and substituting $\posterior^*\InParentheses{\belief^{(*,2)}}=1-\posterior^*\InParentheses{\belief^{(*,1)}}$ condition on $\posterior^*\InParentheses{\belief^{(*,1)}}\in [0,1]$ we get that:
    \begin{align*} 
    &  \expect[\belief\sim \posterior^*] { \Util(\action_1,\belief,\contract)} =  \expect[\belief\sim \posterior^*] {\Util(\action_2,\belief,\contract) }               \\
    \Leftrightarrow \qquad& \sum_{\belief \in \belief^*} \posterior^*(\belief)\sum_{\reward\in\rewards} \belief_{\action_1}(\reward) \InParentheses{ \contract(\reward) -\cost(\action_1)}  =  \sum_{\belief \in \belief^*} \posterior^*(\belief)\sum_{\reward\in\rewards} \belief_{\action_2}(\reward) \InParentheses{ \contract(\reward) -\cost(\action_2)} \\
    \Leftrightarrow \qquad & \posterior^*\InParentheses{\belief^{(*,1)}}\sum_{\reward\in\rewards} \InParentheses{ \contract(\reward) -\cost(\action_1)}\cdot  \belief^{(*,1)}_{\action_1}(\reward)   + \InParentheses{1-\posterior^*\InParentheses{\belief^{(*,1)}}}\sum_{\reward\in\rewards} \InParentheses{ \contract(\reward) -\cost(\action_1)} \belief^{(*,2)}_{\action_1}(\reward) \\
      =  &  \posterior^*\InParentheses{\belief^{(*,1)}}\sum_{\reward\in\rewards} \InParentheses{ \contract(\reward) -\cost(\action_2)}\cdot  \belief^{(*,1)}_{\action_2}(\reward)   + \InParentheses{1-\posterior^*\InParentheses{\belief^{(*,1)}}}\sum_{\reward\in\rewards} \InParentheses{ \contract(\reward) -\cost(\action_2)} \belief^{(*,2)}_{\action_2}(\reward).
    \end{align*}

    The equality above is equivalent to the following equality:
    $$\posterior^*\InParentheses{\belief^{(*,1)}}\cdot D(\contract)= E(\contract) \Leftrightarrow D(\contract)\geq E(\contract) \geq 0 \text{ or } 0 \geq E(\contract) \geq D(\contract),
    $$
    where we used that $\posterior^*\InParentheses{\belief^{(*,1)}}\in[0,1]$.
    \end{proof}

\end{prevproof}

\begin{prevproof}{thm:lp for berk-nash}
    Based on \Cref{lem:charactere actions}, \Cref{lem:union of convex sets}, and substituting $\adist^*(\action_2)=1-\adist^*(\action_1)$, we rewrite the program in \Cref{program:original} as the maximum over the following parametrized quadratic programs for $\widehat{\contracts}\in \InBraces{\contracts_1,\contracts_2}$:
    \begin{align*}
     \max & \quad \adist^*(\action_1)\sum_{\reward\in\rewards} \InParentheses{\reward - \contract(\reward)}\InParentheses{\dist_{\action_1}(\reward) - \dist_{\action_2}(\reward)} + \sum_{\reward\in\rewards}\dist_{\action_2}(\reward) \cdot \InParentheses{\reward - \contract(\reward)}   \\
     \text{s.t.} &\quad   \adist^*(\action_1)\in [\action_1^{min},\action_1^{max}],\contract \in \widehat{\contracts}.
    \end{align*}

    The proof follows by noting that for any fixed contract $\contract\in \widehat{\contracts}$, the maximum value of the parametrized program  is attained at $\action_1^{min}$ if term $\sum_{\reward\in\rewards} \InParentheses{\reward - \contract(\reward)}\InParentheses{\dist_{\action_1}(\reward) - \dist_{\action_2}(\reward)}$ is negative and at $\action_1^{max}$ otherwise. Hence it suffices to solve the parametrized LP in \Cref{program:lp for berk-nash} for $\action_1^{val}\in \InBraces{\action_1^{min},\action_1^{max}}$ and $\widehat{\contracts}\in \InBraces{\contracts_1,\contracts_2}$.
\end{prevproof}

%% file: appendix/unhappy_principal.tex
\subsection{Missing Proofs for Revenue Losses}
\label{apx:revenue_loss}
\begin{prevproof}{lem: KL for unhappy}
The proof of the lemma follows by simple calculations. 
The agent has correctly specified belief for action $a_g$. 
For action $a_b$, we show that 
\begin{align*}
    \KL{\dist_{\action_{b}}}{\belief_{\action_{b}}} =
        p\log\InParentheses{\frac{p}{p-\delta}} \leq  \log(1+2\delta) \leq 2\delta.
\end{align*}
where the first inequality holds since $p\geq \frac{3}{4}$ and $\delta<1-p \leq \frac{1}{4}$.
\end{prevproof}

\begin{prevproof}{lem:correct specification}
Given binary actions, it is well known that the optimal contract is to provide a positive reward only for outcome that maximizes the likelihood ratio. In \cref{instance:unhappy principal}, this corresponds to outcome $2$. 
By offering reward $\frac{c}{\delta}$ for outcome $2$, the agent is incentivized to choose action $a_g$ and generate revenue of $p + 2\delta -c$ for the principal. 
\end{prevproof}

\begin{prevproof}{thm:unhappy principal}
Given misspecified belief $B$, to incentivize the agent to choose action $a_g$, the least costly contract is to provide a positive reward of $\frac{c}{2p-1}$. The expected revenue given this contract is $p + 2 \delta- \frac{c\cdot p}{2p-1}$.

In addition, the principal can choose to provide zero reward to the agent and the agent maximizes his utility by choosing action $\action_b$. The expected revenue in this case is $1-p$. 
\end{prevproof}

%% file: ArXiv.bbl
\begin{thebibliography}{}

\bibitem[Alon et~al., 2021]{AlonDT21}
Alon, T., D{\"{u}}tting, P., and Talgam{-}Cohen, I. (2021).
\newblock Contracts with private cost per unit-of-effort.
\newblock In {\em Proceedings of the {ACM} Conference on Economics and Computation ({EC})}, pages 52--69.

\bibitem[Babichenko et~al., 2016]{BabichenkoPR16}
Babichenko, Y., Papadimitriou, C.~H., and Rubinstein, A. (2016).
\newblock Can almost everybody be almost happy?
\newblock In {\em Proceedings of the 2016 {ACM} Conference on Innovations in Theoretical Computer Science}, pages 1--9.

\bibitem[Castiglioni et~al., 2022]{castiglioni2022designing}
Castiglioni, M., Marchesi, A., and Gatti, N. (2022).
\newblock Designing menus of contracts efficiently: The power of randomization.
\newblock In {\em Proceedings of the {ACM} Conference on Economics and Computation {(EC)}}, pages 705--735.

\bibitem[D{\"{u}}tting et~al., 2021]{DuttingEFK21}
D{\"{u}}tting, P., Ezra, T., Feldman, M., and Kesselheim, T. (2021).
\newblock Combinatorial contracts.
\newblock In {\em Proceedings of the {IEEE} Symposium on Foundations of Computer Science {(FOCS)}}, pages 815--826.

\bibitem[D\"utting et~al., 2023]{DuettingEFK23}
D\"utting, P., Ezra, T., Feldman, M., and Kesselheim, T. (2023).
\newblock Multi-agent contracts.
\newblock In {\em Proceedings of the ACM Symposium on Theory of Computing (STOC)}, pages 1311--1324.

\bibitem[Esponda and Pouzo, 2016]{esponda2016berk}
Esponda, I. and Pouzo, D. (2016).
\newblock Berk--nash equilibrium: A framework for modeling agents with misspecified models.
\newblock {\em Econometrica}, 84(3):1093--1130.

\bibitem[Esponda et~al., 2021]{esponda2021asymptotic}
Esponda, I., Pouzo, D., and Yamamoto, Y. (2021).
\newblock Asymptotic behavior of bayesian learners with misspecified models.
\newblock {\em Journal of Economic Theory}, 195:105260.

\bibitem[Facchinei and Pang, 2007]{facchinei_finite-dimensional_2007}
Facchinei, F. and Pang, J.-S. (2007).
\newblock {\em Finite-dimensional variational inequalities and complementarity problems}.
\newblock Springer Science \& Business Media.

\bibitem[Frick et~al., 2023]{frick2023belief}
Frick, M., Iijima, R., and Ishii, Y. (2023).
\newblock Belief convergence under misspecified learning: A martingale approach.
\newblock {\em The Review of Economic Studies}, 90(2):781--814.

\bibitem[Fudenberg et~al., 2021]{fudenberg2021limit}
Fudenberg, D., Lanzani, G., and Strack, P. (2021).
\newblock Limit points of endogenous misspecified learning.
\newblock {\em Econometrica}, 89(3):1065--1098.

\bibitem[Fudenberg and Levine, 1993]{Fudenberg93self}
Fudenberg, D. and Levine, D.~K. (1993).
\newblock Self-confirming equilibrium.
\newblock {\em Econometrica}, 61(3):523--545.

\bibitem[Fudenberg et~al., 2017]{fudenberg2017active}
Fudenberg, D., Romanyuk, G., and Strack, P. (2017).
\newblock Active learning with a misspecified prior.
\newblock {\em Theoretical Economics}, 12(3):1155--1189.

\bibitem[Grossman and Hart, 1983]{GrossmanHart83}
Grossman, S.~J. and Hart, O.~D. (1983).
\newblock An analysis of the principal-agent problem.
\newblock {\em Econometrica}, 51(1):7--45.

\bibitem[Guruganesh et~al., 2024]{guruganesh2024contracting}
Guruganesh, G., Kolumbus, Y., Schneider, J., Talgam-Cohen, I., Vlatakis-Gkaragkounis, E.-V., Wang, J.~R., and Weinberg, S.~M. (2024).
\newblock Contracting with a learning agent.
\newblock {\em arXiv preprint arXiv:2401.16198}.

\bibitem[Guruganesh et~al., 2021]{guruganesh20}
Guruganesh, G., Schneider, J., and Wang, J. (2021).
\newblock Contracts under moral hazard and adverse selection.
\newblock In {\em Proceedings of the {ACM} Conference on Economics and Computation ({EC})}, pages 563--582.

\bibitem[Holmstr\"om, 1979]{Holmstrom79}
Holmstr\"om, B. (1979).
\newblock Moral hazard and observability.
\newblock {\em The Bell Journal of Economics}, 10(1):74--91.

\bibitem[Hsieh et~al., 2019]{hsieh_convergence_2019}
Hsieh, Y., Iutzeler, F., Malick, J., and Mertikopoulos, P. (2019).
\newblock On the convergence of single-call stochastic extra-gradient methods.
\newblock In {\em Annual Conference on Neural Information Processing Systems}.

\bibitem[Innes, 1990]{innes1990limited}
Innes, R.~D. (1990).
\newblock Limited liability and incentive contracting with ex-ante action choices.
\newblock {\em Journal of economic theory}, 52(1):45--67.

\bibitem[Korpelevich, 1976]{korpelevich_extragradient_1976}
Korpelevich, G.~M. (1976).
\newblock The extragradient method for finding saddle points and other problems.
\newblock {\em Matecon}, 12:747--756.

\bibitem[Li and Pei, 2020]{li2020misspecified}
Li, Y. and Pei, H. (2020).
\newblock Misspecified beliefs about time lags.
\newblock {\em arXiv preprint arXiv:2012.07238}.

\bibitem[Li and Slivkins, 2022]{li2022incentivizing}
Li, Y. and Slivkins, A. (2022).
\newblock Incentivizing participation in clinical trials.
\newblock {\em arXiv preprint arXiv:2202.06191}.

\bibitem[Mertikopoulos et~al., 2018]{MertikopoulosPP18}
Mertikopoulos, P., Papadimitriou, C.~H., and Piliouras, G. (2018).
\newblock Cycles in adversarial regularized learning.
\newblock In Czumaj, A., editor, {\em Proceedings of the Twenty-Ninth Annual {ACM-SIAM} Symposium on Discrete Algorithms, {SODA} 2018, New Orleans, LA, USA, January 7-10, 2018}, pages 2703--2717. {SIAM}.

\bibitem[Nyarko, 1991]{nyarko1991learning}
Nyarko, Y. (1991).
\newblock Learning in mis-specified models and the possibility of cycles.
\newblock {\em Journal of Economic Theory}, 55(2):416--427.

\bibitem[Popov, 1980]{popov_modification_1980}
Popov, L.~D. (1980).
\newblock A modification of the {Arrow}-{Hurwicz} method for search of saddle points.
\newblock {\em Mathematical notes of the Academy of Sciences of the USSR}, 28(5):845--848.

\bibitem[Rubinstein, 2016]{Rubinstein16}
Rubinstein, A. (2016).
\newblock Settling the complexity of computing approximate two-player nash equilibria.
\newblock In {\em {IEEE} 57th Annual Symposium on Foundations of Computer Science, {FOCS}}, pages 258--265.

\end{thebibliography}
